%% file: AAMAS_2026_camera_ready.tex
\documentclass[sigconf]{aamas}

\usepackage{balance} 

\usepackage{multirow}
\usepackage[export]{adjustbox}
\usepackage{graphbox,graphicx}

\usepackage{amsmath}
\usepackage{mathtools}
\usepackage{amsthm}
\usepackage{graphicx}
\usepackage{subfigure,verbatim}
\usepackage{wrapfig}

\def\EQi{\overline{EQ}^i}
\RequirePackage{algorithm}
\RequirePackage{algorithmic}

\theoremstyle{plain}
\newtheorem{theorem}{Theorem}[section]

\theoremstyle{definition}
\newtheorem{definition}{Definition}[section]

\theoremstyle{remark}

\doi{OBIA4073}

\makeatletter
\gdef\@copyrightpermission{
  \begin{minipage}{0.2\columnwidth}
   \href{https://creativecommons.org/licenses/by/4.0/}{\includegraphics[width=0.90\textwidth]{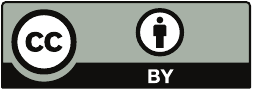}}
  \end{minipage}\hfill
  \begin{minipage}{0.8\columnwidth}
   \href{https://creativecommons.org/licenses/by/4.0/}{This work is licensed under a Creative Commons Attribution International 4.0 License.}
  \end{minipage}
  \vspace{5pt}
}
\makeatother

\setcopyright{ifaamas}
\acmConference[AAMAS '26]{Proc.\@ of the 25th International Conference
on Autonomous Agents and Multiagent Systems (AAMAS 2026)}{May 25 -- 29, 2026}
{Paphos, Cyprus}{C.~Amato, L.~Dennis, V.~Mascardi, J.~Thangarajah (eds.)}
\copyrightyear{2026}
\acmYear{2026}
\acmDOI{}
\acmPrice{}
\acmISBN{}

\title[Generalized Per-Agent Advantage Estimation for Multi-Agent Policy Optimization]{Generalized Per-Agent Advantage Estimation for Multi-Agent Policy Optimization}

\author{Seongmin Kim}
\orcid{0009-0008-3032-8943}
\affiliation{
  \institution{KAIST}
  \city{Daejeon}
  \country{Korea}}
\email{seongmin.kim@kaist.ac.kr}

\author{Giseung Park}
\orcid{0000-0002-9737-4142}
\affiliation{
  \institution{University of Toronto}
  \city{Toronto, Ontario}
  \country{Canada}}
\email{giseung.park@utoronto.ca}

\author{Woojun Kim}
\orcid{0000-0001-8886-8848}
\affiliation{
  \institution{Carnegie Mellon University}
  \city{Pittsburgh, Pennsylvania}
  \country{USA}}
\email{woojunk@andrew.cmu.edu}

\author{Jiwon Jeon}
\orcid{0009-0005-4255-5412}
\affiliation{
  \institution{KAIST}
  \city{Daejeon}
  \country{Korea}}
\email{jiwon.jeon@kaist.ac.kr}

\author{Seungyul Han}
\orcid{0000-0002-5376-1976}
\affiliation{
  \institution{UNIST}
  \city{Ulsan}
  \country{Korea}}
\email{syhan@unist.ac.kr}
\authornote{Corresponding author}

\author{Youngchul Sung}
\orcid{0000-0003-4536-6690}
\affiliation{
  \institution{KAIST}
  \city{Daejeon}
  \country{Korea}}
\email{ycsung@kaist.ac.kr}

\begin{abstract}
In this paper, we propose a novel framework for multi-agent reinforcement learning that enhances sample efficiency and coordination through accurate per-agent advantage estimation. The core of our approach is Generalized Per-Agent Advantage Estimator (GPAE), which employs a per-agent value iteration operator to compute precise per-agent advantages. This operator enables stable off-policy learning by indirectly estimating values via action probabilities, eliminating the need for direct $Q$-function estimation. To further refine estimation, we introduce a double-truncated importance sampling ratio scheme. This scheme improves credit assignment for off-policy trajectories by balancing sensitivity to the agent’s own policy changes with robustness to non-stationarity from other agents. Experiments on benchmarks demonstrate that our approach outperforms existing approaches, excelling in coordination and sample efficiency for complex scenarios.
\end{abstract}


\newcommand{\BibTeX}{\rm B\kern-.05em{\sc i\kern-.025em b}\kern-.08em\TeX}

\begin{document}


\pagestyle{fancy}
\fancyhead{}


\maketitle 


\input{1.introduction}
\input{2.background}
\input{3.motivation}
\input{4.method}

\input{5.experiments}
\input{6.relatedwork}

\input{7.discussion}

\begin{acks}
This work was supported by Center for Applied Research in Artificial Intelligence (CARAI) Grant funded by Defense Acquisition Program Administration (DAPA) and Agency for Defense Development (ADD) of Republic of Korea (UD230017TD, 50\%), and the National Research Foundation of Korea(NRF) grant funded by the Korea government(MSIT) (No. RS-2025-00557589, Generative Model Based Efficient Reinforcement Learning Algorithms for Multi-modal Expansion in Generalized Environments, 50\%).
\end{acks}



\bibliographystyle{ACM-Reference-Format} 
\bibliography{sample}

\newpage
\appendix
\input{appendix}


\end{document}

%% file: 1.introduction.tex
\section{Introduction}

Multi-agent reinforcement learning (MARL) addresses cooperative tasks where agents learn policies to maximize shared rewards through environmental interactions \cite{marl_survey}. These tasks face challenges from coordination requirements and partial observability. While independent learning approaches \cite{indepmarl} often suffer from non-stationarity issues, the centralized training and decentralized execution (CTDE) paradigm \cite{ctde, maddpg} enables agents to utilize global information during training while maintaining local observation dependency during execution.

Within CTDE, various algorithms have emerged, including value-based methods like VDN and QMIX \cite{vdn, QMIX} and multi-agent policy gradient~(MAPG)-based methods such as COMA and MAPPO \cite{coma, mappo}. Among these, MAPPO, an extension of Proximal Policy Optimization (PPO, \citet{ppo}), has gained popularity for its strong performance. However, MAPPO's reliance on Generalized Advantage Estimator (GAE) \cite{gae} introduces limitations: (1) reduced sample efficiency due to on-policy constraints, and (2) inadequate credit assignment from estimating identical advantages for all agents. While techniques like V-traces enhance sample efficiency in single-agent settings, their application to multi-agent systems has not been theoretically or experimentally explored. The second problem is that MAPPO updates policies based on centralized critic with GAE, which estimates the same advantage for all individual agents' actions in the same timestep. This approach is inadequate for solving the multi-agent credit assignment problem, which is a significant challenge in MARL where the contribution of individual agents to the global reward must be accurately assessed.

The \textit{multi-agent credit assignment problem} significantly impacts MARL performance \cite{lica, mappg}. Value-based methods address this with value function factorization \cite{QMIX,qtran}, but lack theoretical guarantees and often fail to provide accurate credit assignment in complex tasks. In MAPG area, COMA employ counterfactual baselines but suffer from variance and scalability issues. DAE \cite{dae} offers a GAE-based approach but lacks policy-invariance \cite{ng99}, potentially leading to sub-optimal convergence.

Motivated by these limitations, we investigate an advantage estimator for ensuring precise per-agent credit assignment at $n$-step horizons, remains policy-invariant, and enables stable off-policy reuse under MAPG within CTDE.

The key contributions of our work are summarized as follows:

\begin{itemize}
\item We propose \emph{Generalized Per-Agent Advantage Estimator}, \\which provides explicit per-agent credit signals under CTDE and unifies on-policy learning and off-policy reuse within a single estimator.
\item We establish contraction of the per-agent operator and show policy invariance of GPAE theoretically.
\item We propose a \emph{double-truncated} importance sampling scheme tailored to multi-agent coupling, improving stability and maintaining credit fidelity compared to single-truncation or individual-only truncation.
\item We demonstrate the practical benefits of our method through extensive experiments, showing significant performance improvements in MARL tasks.
\end{itemize}

%% file: 2.background.tex
\section{Background}
\label{sec:background}
\subsection{Dec-POMDP and CTDE Setup}
A multi-agent system can be modeled as a Decentralized Partially Observable Markov Decision Process (Dec-POMDP) \cite{Dec-POMDP}, defined by the tuple $\langle \mathcal{S}, \mathcal{A}, P, R, \mathcal{Z}, O, n, \gamma \rangle$, where $\mathcal{S}$ is the global state space, $\mathcal{A} = \prod_{i=1}^n \mathcal{A}^i$ is the joint action space of $n$ agents, with $\mathcal{A}^i$ being the action space of agent $i$, $P(s'|s, \boldsymbol{a})$ is the state transition probability given the current state $s$ and joint action $\boldsymbol{a} = \{a^1, \dots, a^n\}$, $R(s, \boldsymbol{a})$ is the shared reward function, $\mathcal{Z}$ is the observation space, and $O(o^i|s, i)$ specifies the observation probabilities for each agent $i$, and $\gamma \in [0, 1)$ is the discount factor.

The goal of MARL is to learn a policy  $\boldsymbol{\pi} = \prod_{i=1}^n \pi^i(a^i | o^i)$  for each agent that maximizes the expected cumulative reward  $J(\boldsymbol{\pi}) = \mathbb{E}_{\boldsymbol{\pi}} \left[\sum_{t=0}^\infty \gamma^t R(s_t, \boldsymbol{a}_t)\right]$. The CTDE paradigm facilitates MARL by allowing agents to use global information including full state  $s$ during training, while restricting each agent’s policy to depend only on local observations  $o_i$  during execution. Note that our theoretical analysis assumes full observability during training, which is common in many CTDE-based MARL studies \citep{ob,more} to enable rigorous analysis. Extending theory to partially observable settings remains challenging, as solving Dec-POMDPs is NEXP-complete and requires super-exponential time in the worst case~\citep{bernstein2002complexity, marl2}. 

\subsection{Advantage Estimation and Credit Assignment}
In MAPG methods, accurate advantage estimation plays a central role in stabilizing policy updates and enabling effective credit assignment. Among existing estimators, GAE~\cite{gae} is widely adopted for its ability to balance bias and variance through an exponentially-weighted sum of temporal difference (TD) errors with bias-variance trade-off parameter $\lambda$. 

MAPPO~\cite{mappo}, a strong baseline in cooperative MARL, extends PPO to the multi-agent setting under CTDE. It applies GAE as the default advantage estimator for each agent. However, MAPPO treats all agents identically by assuming a shared team advantage:
\begin{equation}
\label{eq:mappo_advantage}
    A^i(s, \mathbf{a}) = A^\text{global}(s, \mathbf{a}) \quad \forall i,
\end{equation}
which ignores the agent-specific impact on the joint outcome and thus limits the fidelity of credit assignment. To address this limitation, COMA~\cite{coma} proposes a per-agent advantage formulation using a counterfactual baseline that marginalizes out the individual agent’s contribution. COMA computes a per-agent advantage using a counterfactual baseline as:
\begin{equation}
    A^{i,\text{COMA}}(s, \mathbf{a}) = Q(s, \mathbf{a}) - \sum_{a^i} \pi_i(a^i | o^i) Q(s, (a^i, \mathbf{a}^{-i})),
\end{equation}
where $\mathbf{a}^{-i}$ represents the actions of all agents except $i$. This formulation enables per-agent credit assignment by marginalizing over agent $i$’s actions.

More recently, Difference Advantage Estimation (DAE)~\cite{dae} was developed as a GAE-based alternative by incorporating potential-based difference rewards to address multi-step credit assignment:
\begin{equation}
     A^{i,\text{DAE}}_t = \sum_{l=0}^\infty (\gamma \lambda)^l \delta^i_{t+l},
\end{equation}
where $\delta^i_{t+l} = r_{t+l}- \beta^{l+1}\mathbb{E}_{a^i}[r_{t+l}]+\gamma V_{t+l+1}-V_{t+l}$ with bias-credit assignment control parameter $\beta$.

%% file: 3.motivation.tex
\section{Motivation}
\label{sec:motivation}

MARL poses unique challenges, among which the multi-agent credit assignment problem is a key bottleneck. Accurately assigning credit to individual agents based on their contributions to the global reward is critical for effective learning. However, existing methods in MAPG domain exhibit fundamental limitations.

Building on the background and Eq.~\eqref{eq:mappo_advantage}, we note that MAPPO’s reliance on a single advantage function for every agent can mask each agent’s distinct contributions. 
COMA advantage relies on TD$(0)$ estimation, which limits its ability to leverage $n$-step estimation and $n$-step credit-assignment.
DAE~\cite{dae} extends GAE with a potential-based difference reward to address credit assignment. While DAE allows a $n$-step credit assignment, it is not policy invariant due to the estimated reward bias term, which can lead to suboptimal results. Additionally, DAE suffers from instability in explicit reward estimation.

Beyond credit assignment, sample inefficiency further complicates learning in MAPG. Importance sampling offers a theoretical means to utilize off-policy data, but high variance under MARL often reduces its practical utility. While single-agent methods like V-trace \cite{vtrace} control variance by truncating importance sampling ratios, their effectiveness does not readily transfer to multi-agent systems because each agent’s behavior interdepends on others. 

We thus seek an advantage estimator that (i) delivers per-agent, $n$-step credit signals under CTDE, (ii) remains policy-invariant to avoid bias, and (iii) is compatible with off-policy data. Table~\ref{tab:advantage_comparison} summarizes how existing MAPG estimators fall short with respect to these criteria.

\begin{table}[!h]
\small
\caption{Comparison of MAPG advantage estimators}
    \begin{tabular}{l|c|c|c}
        \toprule
        \multirow{2}{*}{Advantage} & Credit & \multirow{2}{*}{Off-Policy} & Additional \\
        & Assignment && Bias($\lambda=1$)\\
        \midrule
        GAE & No & No & No \\
        COMA & 1-step & No & No \\
        DAE & \(n\)-step & No & Yes \\
        \textbf{GPAE (Ours)} & \textbf{\(\boldsymbol{n}\)-step} & \textbf{Yes} & \textbf{No} \\
        \bottomrule
    \end{tabular}
    \label{tab:advantage_comparison}
\end{table}

Guided by these requirements, we introduce GPAE, which provides $n$-step, per-agent credit signals while preserving policy invariance and enabling off-policy reuse. 

\subsection{Credit Assignment in the Presence of Anomalous Behavior}
\label{subsec:motivation_exp}

Before detailing the methodology, we provide some preliminary experimental results in this section to reinforce our motivation. Here, we examine how various advantage estimators respond when an agent deviates from cooperation.

For the experiment, we use the SMAX~\cite{jaxmarl} environment, a JAX-implemented version of SMAC~\cite{smac}, which is a widely used benchmark for MARL algorithms. In SMAX, agents collaborate to achieve a common goal, typically involving cooperation in battle against enemy. In the SMAX-3m task, three agents cooperate to win a battle, requiring careful coordination and credit assignment.

To demonstrate how  methods handle the multi-agent credit assignment, we modify the SMAX-3m task by introducing  5\% probability that one agent performs a \emph{stop} action. This irregularity disrupts coordination, underscoring the need for robust credit assignment and stable advantage estimation. 

We focus on whether each advantage estimator can accurately penalize the “misbehaving” agent’s suboptimal action and thus guide the overall team to learn more robust behaviors. In particular, we track the difference in per-agent advantage at the timestep where the anomalous action occurs. Formally, we measure:
\[
    \Delta A = \frac{1}{N-1}\sum_{j \neq i} \hat{A}_t^j - \hat{A}_t^i,
\]
where \( \hat{A}_t^i \) is the per-agent advantage for the misbehaving agent $i$ at timestep \( t \), and \( \hat{A}_t^j \) are the advantages of the other \( N-1 \) agents. A higher \(\Delta A\) indicates that the advantage more strongly penalizes the anomalous agent relative to agents acting optimally, implying better credit assignment. 

\begin{figure}[h]
\centering
    
    \subfigure[Advantage Gap $\Delta A$]{
    \centering\includegraphics[clip, trim=0 0 360 0, width=0.53\columnwidth]{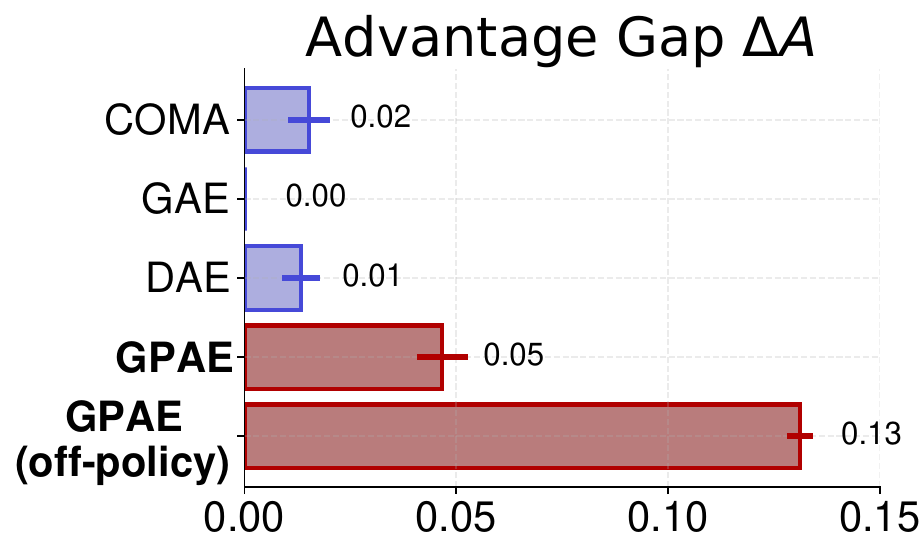}
        \label{subfig:motivation_advantage_gap}
    }
    \subfigure[Performance]{
    \centering\includegraphics[clip, trim=450 0 0 0, width=0.42\columnwidth]{figures/adv_motivation/gap_and_winrate_barh.pdf}
        \label{subfig:motivation_performance}
    }
\caption{\subref{subfig:motivation_advantage_gap}~Advantage gap \(\Delta A\) indicating how strongly each method penalizes the anomalous agent. Higher \(\Delta A\) means more effective credit assignment against the “stop” action. 
\subref{subfig:motivation_performance}~Average win rates, illustrating overall learning stability and performance.}
\Description{
(a) Advantage gap \(\Delta A\) indicating how strongly each method penalizes the anomalous agent. Higher \(\Delta A\) means more effective credit assignment against the “stop” action. 
(b) Average win rates, illustrating overall learning stability and performance.
}
\label{fig:motivation} 
\end{figure}

Figure~\ref{subfig:motivation_advantage_gap} compares the \emph{advantage gap} \(\Delta A\) of several advantage estimators. Our proposed GPAE yields the highest \(\Delta A\) (particularly in its off-policy form), demonstrating its ability to measure credit assignment effectively, while COMA and DAE achieves a nonzero gap but are less effective in credit assignment. Note that GAE using common advantage for all agents cannot extract advantage gap. 
Figure~\ref{subfig:motivation_performance} shows the average win rate, assessing the overall performance. For this figure, we used the same PPO-style policy optimization only with advantage changed.

Note that COMA performs significantly worse than other advantage estimation methods even though it extracts larger advantage gap than GAE and DAE. The main difference of COMA from other methods is that COMA uses one-step estimation for advantage, which is prone to large errors,  whereas all others use $n$-step advantage estimation. Hence, we recognize that the accuracy of estimation by $n$-step estimation as well as credit assignment is important for performance. It is seen that under the $n$-step estimation framework, the advantage gap, i.e., credit assignment indeed translates to the overall performance and its impact saturates as the advantage gap increases, following our intuition. 

These observations motivate our approach to unify on-policy and off-policy advantage estimation around $n$-step signals while preserving per-agent credit assignment. In the following sections, we introduce GPAE and the associated DT-ISR weighting, along with theoretical properties and empirical validation.

%% file: 4.method.tex
\section{Methodology}
\label{sec:gpae}

In this section, we introduce two contributions of our work. First, based on our value iteration operator, we propose a novel advantage estimator, called \textbf{G}eneralized \textbf{P}er-agent \textbf{A}dvantage \textbf{E}stimator (\textbf{GPAE}). Then we extend the method to the off-policy case with theoretical insights, and propose a \textit{double-truncated importance sampling} for the off-policy correction term in the weighted advantage.

\subsection{Generalized Per-Agent Advantage Estimation}

To begin, we consider the counterfactual value function \\$\mathbb{E}_{a^i \sim \pi^i}[Q(s, a^i, \boldsymbol{a}^{-i})]$,
where $Q$ denotes the joint action-value function estimate $Q^{\boldsymbol{\pi}}$, and $\boldsymbol{a}^{-i}$ represents the actions of all agents except agent $i$. For convenience, we define:
\[
\overline{EQ}^i := \mathbb{E}_{a^i \sim \pi^i}[Q(s, a^i, \boldsymbol{a}^{-i})].
\]
Here, $\overline{EQ}^i$ serves as an estimate of the \emph{per-agent state-value function}, which averages out only the action of agent $i$ under its policy $\pi^i$ and thus has $(s,\boldsymbol{a}^{-i})$ as its argument, in contrast to the full state-value function $V(s) = \mathbb{E}_{\boldsymbol{a} \sim \boldsymbol{\pi}}[Q(s, \boldsymbol{a})]$ having only $s$ as the argument.

To build an efficient way to update $\overline{EQ}^i$, we present a new per-agent value iteration operator $\mathcal{R}^i$, a specifically designed for MARL to facilitate effective credit assignment by focusing on per-agent value estimation. For now, it only considers on-policy learning, thus we denote the operator with $\mathcal{R}_{\text{on}}^i$.

\begin{definition}
    \label{def:Ri_on}
The value iteration operator $\mathcal{R}_{\text{on}}^i$ for $\overline{EQ}^i$ is defined as 
\begin{multline}
\label{eq:Ri_on}
\mathcal{R}_{\text{on}}^i\overline{EQ}^i := \overline{EQ}^i + 
\mathbb{E}_{a_0^i \sim \pi^i}\Biggl[\mathbb{E}_{\boldsymbol{\pi}}\biggl[\sum_{t \geq 0} (\gamma\lambda)^t \big(r_t + \gamma \overline{EQ}_{t+1}^i - \overline{EQ}_t^i\big) \\
\biggr| (s_0, a_0^i, \boldsymbol{a}_0^{-i})=(s,a^i,\boldsymbol{a}^{-i})\biggr]\Biggr],
\end{multline}
where  $\overline{EQ}^i_t := \mathbb{E}_{a^i_t\sim \pi^i}[Q(s_t,a_t^i,\boldsymbol{a}^{-i}_t)]$ and $\lambda \in (0,1]$ is a bias-variance trade-off parameter.
\end{definition}

$\mathcal{R}_{\text{on}}^i$ introduces the ability to compute per-agent state-value estimates efficiently, incorporating partial averaging over the actions of the agent $i$ while retaining dependencies on the actions of other agents, $\boldsymbol{a}^{-i}$. This structure enables $n$-step credit assignment from the value estimation step, in contrast to existing methods that consider only 1-step credit assignment and value estimation separately. Our proposed method provides a more stable and efficient learning method with estimation by connecting advantage and value in a novel way.

The operator $\mathcal{R}_{\text{on}}^i$ possesses desirable theoretical properties. We establish the contraction property of $\mathcal{R}^i$, which ensures convergence to a unique fixed point:
\begin{theorem}
\label{thm:contraction_on}
$\mathcal{R}_{\text{on}}^i$ is a $\gamma$-contraction, which means that $\overline{EQ}^i$ converges to the unique fixed point. If $~\lambda = 1$, the fixed point is \\$\mathbb{E}_{a^i \sim \pi^i}[Q^{\boldsymbol{\pi}}(s, a^i, \boldsymbol{a}^{-i})]$. 
\end{theorem}
\begin{proof}[Proof Sketch]
Rewriting $R^i_{\text{on}}$ by shifting the TD-error indices shows that the
difference between two inputs $\overline{EQ}^{i,1}$ and $\overline{EQ}^{i,2}$ is scaled by a
geometric factor $(\gamma\lambda)^t$.  
Taking the sup-norm gives an overall
contraction of at most~$\gamma$, so $R^i_{\text{on}}$ is a $\gamma$-contraction.
For $\lambda = 1$, the TD terms telescope, leaving
$[Q^{\boldsymbol{\pi}}(s,\boldsymbol{a}) - \mathbb{E}_{a^i \sim \pi^i}[Q(s, a^i, \boldsymbol{a}^{-i})]$, and the fixed point
becomes $\mathbb{E}_{a^i \sim \pi^i}[Q^{\boldsymbol{\pi}}(s, a^i, \boldsymbol{a}^{-i})]$.
\end{proof} 

Theorem \ref{thm:contraction_on} states that $\mathbb{E}_{a^i \sim \pi^i}[Q^{\boldsymbol{\pi}}(s, a^i, \boldsymbol{a}^{-i})]$ can be estimated directly from the policy distribution of agent $i$ and the actions of the remaining agents, without the dual process of estimating the true joint $Q$ value and marginalizing over the policy at $Q$.

From the definition \ref{def:Ri_on} to compute the $n$-step advantage with per-agent TD error without compromising the structure of the operator $\mathcal{R}_{\text{on}}^i$, the following theorem is presented:
\begin{theorem}
\label{thm:invariance_on}
The inner conditional expectation term of $\mathcal{R}^i$, i.e., 
$\mathbb{E}_{\boldsymbol{\pi}}[\sum_{t \geq 0} (\gamma\lambda)^t (r_t + \gamma \overline{EQ}_{t+1}^i - \overline{EQ}_t^i)]$ given  $(s_0, a_0^i, \boldsymbol{a}_0^{-i})=(s,a^i,\boldsymbol{a}^{-i})$
reduces to $Q^{\boldsymbol{\pi}}(s, \boldsymbol{a}) - \mathbb{E}_{a^i \sim \pi^i}[Q(s, a^i, \boldsymbol{a}^{-i})]$, ensuring policy invariance with $\lambda = 1$. 
\end{theorem}
\begin{proof}[Proof Sketch]
With $\lambda=1$, the sum
$\sum_t (r_t + \gamma \overline{EQ}_{t+1}^i - \overline{EQ}_t^i)$
collapses to
$Q^{\boldsymbol{\pi}}(s, \boldsymbol{a}) - \mathbb{E}_{a^i \sim \pi^i}[Q(s, a^i, \boldsymbol{a}^{-i})]$ by telescoping.
Since the baseline term involving $\overline{EQ}^i$ has zero gradient under
decentralized execution, this advantage yields the correct policy
gradient.
\end{proof}
Theorem \ref{thm:invariance_on} suggests that these $n$-step TD-errors can be utilized to achieve unbiased policy updates.
In general, The trade-off between variance and bias can be controlled by adjusting $\lambda$.

Using Theorem \ref{thm:invariance_on}, we propose the GPAE for on-policy learning:
\begin{equation}
\label{eq:gpae_on}
\hat{A}_t^{i,\text{GPAE,on}} := \sum_{l \geq t} (\gamma\lambda)^{l-t} \delta_l^{i,\text{GPAE}},
\end{equation}
where $\delta_l^{i,\text{GPAE}} = r_l + \gamma \overline{EQ}_{l+1}^i - \overline{EQ}_l^i$ is the per-agent TD error.

In the single-agent case, GPAE reduces to GAE(\( \lambda \)). This demonstrates GPAE as a generalization of GAE(\( \lambda \)) for multi-agent scenarios with explicit credit assignment.

\subsection{Extension to Off-Policy Estimation}

While $\mathcal{R}^i$ and GPAE were originally defined for on-policy learning, adapting them to off-policy scenarios requires addressing the mismatch between the behavior policy $\mu^i$ and the target policy $\pi^i$. To achieve this, we reformulate $\mathcal{R}_{\text{on}}^i$ to include importance sampling ratio (ISR) weights, allowing the estimation of per-agent value functions and advantages in an off-policy setting.
\begin{definition}
The generalized per-agent operator $\mathcal{R}^i$ is now defined as:
\begin{multline}
\label{eq:Ri_off}
\mathcal{R}^i\overline{EQ}^i := \overline{EQ}^i + 
\mathbb{E}_{a_0^i \sim \mu^i}\Biggl[\rho_0^i \mathbb{E}_{\boldsymbol{\mu}}\biggl[\sum_{t \geq 0} \gamma^t \big(\prod_{j=1}^t c_j^i\big) \cdot \\
\big(r_t + \gamma \overline{EQ}_{t+1}^i - \overline{EQ}_t^i\big)\biggr| s_0, a_0^i, \boldsymbol{a}_0^{-i}\biggr]\Biggr],
\end{multline}
where $(s_0, \boldsymbol{a}_0) = (s, \boldsymbol{a})$, $\rho_0^i = \frac{\pi^i(a_0^i|o_0^i)}{\mu^i(a_0^i|o_0^i)}$, $\boldsymbol{\mu} = \prod_{i=1}^N \mu^i$ is the behavior policy, and $c_t^i$ is the trace weight for off-policy correction. For the case of on-policy learning, $c_t^i$ is substituted for the value of $\lambda$, a bias-variance trade-off parameter.
\end{definition}
This operator introduces the necessary corrections to the per-agent value updates by weighting the contributions of sampled transitions according to their likelihood under the target policy. The theoretical properties of $\mathcal{R}^i$, including its contraction property, remain intact under the condition that the behavior policy sufficiently explores the state-action space.

The off-policy GPAE is constructed analogously by incorporating ISR weights into the per-agent TD error. Finally, our off-policy generalized per-agent advantage estimator is given by:
\begin{equation}
\label{eq:gpae_off}
\hat{A}_t^{i,\text{GPAE}} := \sum_{l \geq t} \gamma^{l-t} \bigg(\prod_{j=t+1}^l c_j^i\bigg) \delta_l^{i,\text{GPAE}},
\end{equation}
where $\delta_l^{i,\text{GPAE}} = r_l + \gamma \overline{EQ}_{l+1}^i - \overline{EQ}_l^i$ remains the per-agent TD error.

Note that the per-agent value iteration operator denoted by $\mathcal{R}^i$ and $\hat{A}_t^{i,\text{GPAE}}$ expanded to off-policy learning, are the broader framework that encompass on-policy learning case. That is to say, $\mathcal{R}^i$ and $\hat{A}_t^{i,\text{GPAE}}$ in Eq.~\ref{eq:Ri_off} and \ref{eq:gpae_off} inherently transform into $\mathcal{R}^i_{\text{on}}$ and $\hat{A}_t^{i,\text{GPAE,on}}$ in Eq.~\ref{eq:Ri_on} and \ref{eq:gpae_on} in the on-policy scenario, without the necessity of additional modifications.

To ensure the validity of the off-policy extension, we establish the following theoretical guarantees:
\begin{theorem}
\label{thm:contraction_off}
$\mathcal{R}^i$ is a $\gamma$-contraction in the general off-policy case. $\hat{A}_t^{i,\text{GPAE}}$ also guarantees policy invariance with $c_t^i = \boldsymbol{\rho}_t$. 
\end{theorem}
\begin{proof}[Proof Sketch]
In the off-policy operator $R^i$, each TD term is weighted by
$c^i_t\in[0,1]$.  
Bounding these weights shows that the operator still scales
differences by at most~$\gamma$, preserving the $\gamma$-contraction.
When $c_{i,t}=\rho_t$ (full IS ratios), all bias terms vanish and the
expression matches the on-policy case \cite{Precup2000}, giving the same
$Q^\pi - \mathbb{E}_{a_i}Q$ form.  
Thus the estimator remains unbiased.
\end{proof}
While Theorem~\ref{thm:contraction_off} establishes the theoretical foundations, practical implementation requires addressing variance explosion due to the use of off-policy samples. This phenomenon, arising from the mismatch between behavioral and target policies, necessitates introducing a bias into the ISR weight to stabilize learning. In single-agent settings, techniques like V-trace have been effective in achieving this balance. However, applying V-trace directly to multi-agent systems introduces new challenges, which we address in the following subsection.



\subsection{Choice of Off-Policy Correction Term \texorpdfstring{$c^i$}{Lg}}
\label{subsec:dt}

To stabilize off-policy estimation in MARL, we extend the idea of truncated importance sampling ratio (ISR). In single-agent methods, the weight is defined as $c_t = \min (1, \rho_t)$, which bounds variance by truncating the ISR. Directly applying this to multi-agent systems, however, fails to account for the interaction between agents: a shared joint ratio controls variance but obscures individual contributions, whereas using only agent-specific ratios preserves credit signals but can destabilize training. To address this, we introduce a new double-truncated ISR weight (DT-ISR) that balance variance control with per-agent credit sensitivity.

We first consider two direct approaches to obtain weights in multi-agent systems. Single truncation (ST) applies a shared truncation on a joint ISR $\boldsymbol{\rho} = \prod_{j \in \mathcal{N}} \pi^j_t/\mu^j_t$, defining $c_t^{i,\text{ST}} = \lambda \min\left(1, \boldsymbol{\rho}\right)$ for every agent $i$. ST effectively controls variance under non-stationary team dynamics because all agents receive the same bounded weight, but it makes the update insensitive to an individual agent’s policy change and thus weakens per-agent credit signals.

As opposed to ST, we can only consider the ISR of the own policy: $c_t^{i,\text{IT}} = \lambda \min\left(1, \rho_t^i\right),$ where $\rho_t^i = \pi^i_t/\mu^i_t$ and IT denotes individual truncation. This can be effective to focus on own individual policy, but it ignores shifts induced by the team. As a result, the training signal is evaluated under a mismatched occupancy. This leads to unstable learning due to biased updates since samples collected under outdated team behavior are not sufficiently discounted.

Therefore, it is essential to propose a novel truncation scheme for MARL that incorporates the strengths of both IT and ST. The scheme should be constrained by a bound of $1$ while maintaining a close alignment with both individual ISR $\rho_t^i$ and joint ISR $\boldsymbol{\rho}$.

In this regard, we propose a novel truncation scheme for the trace coefficient of multi-agent off-policy correction: a \textit{\textbf{double-truncated importance sampling ratio (DT-ISR) weight}}:
\begin{equation}
\label{eq:cDTit}
c_t^{i,\text{DT}} = \lambda \min\left(1, \rho_t^i \min(\eta, \boldsymbol{\rho}_t^{-i})\right),
\end{equation}
where $\boldsymbol{\rho}_t^{-i} = \prod_{j \neq i} \pi^j_t/\mu^j_t$, and $\eta$ is a constant that serves to mitigate the effects of other agents' policies, while making the unique ISR of agent $i$ less intrusive. This formulation ensures that the unique contribution of agent $i$ is preserved while mitigating variance introduced by the collective policies of other agents. DT-ISR provides a better approximation of the ISR weight, enabling more effective off-policy value and advantage estimation with improved stability. 

\begin{figure}[ht]
    \centering
    \includegraphics[width=0.75\linewidth]{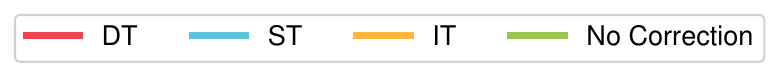}
    
    \subfigure[Distance from $\rho^i$]{
        \includegraphics[width=0.48\linewidth]{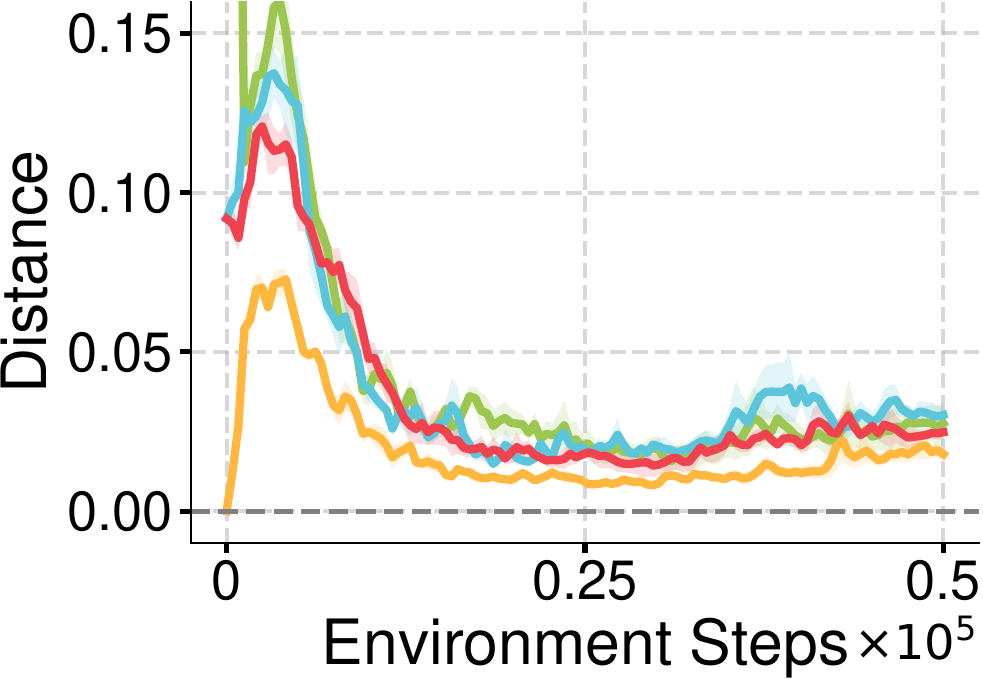}%
        \label{subfig:dt_diff1}
    }%
    \subfigure[Distance from $\boldsymbol{\rho}$]{
        \includegraphics[width=0.48\linewidth]{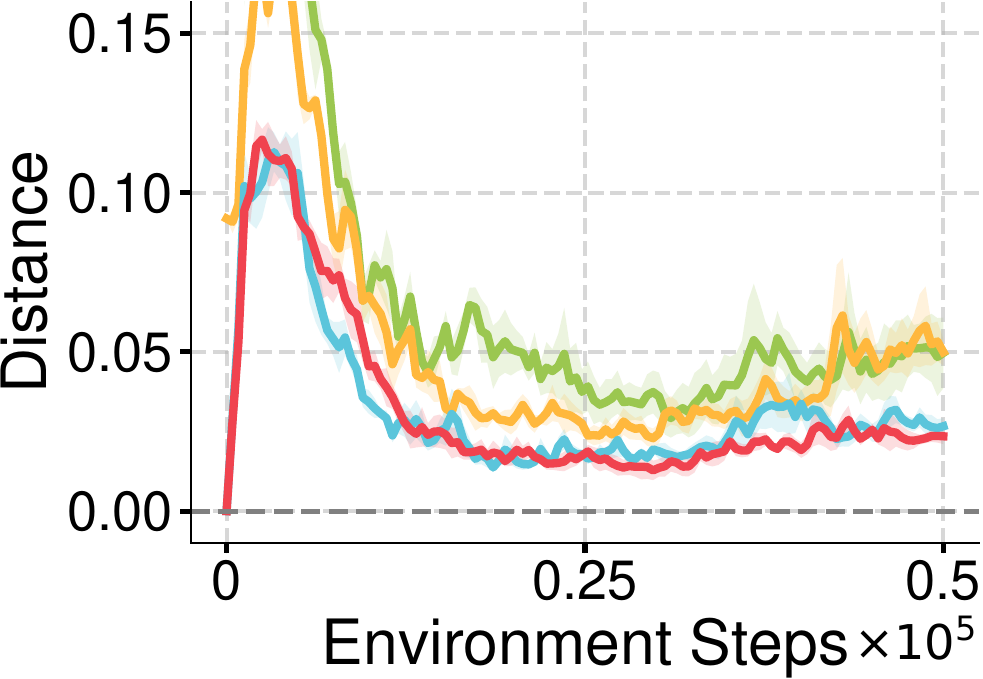}
        \label{subfig:dt_diff2}
    }
    
    \subfigure[Gap between (a) and (b)]{
        \includegraphics[width=0.48\linewidth]{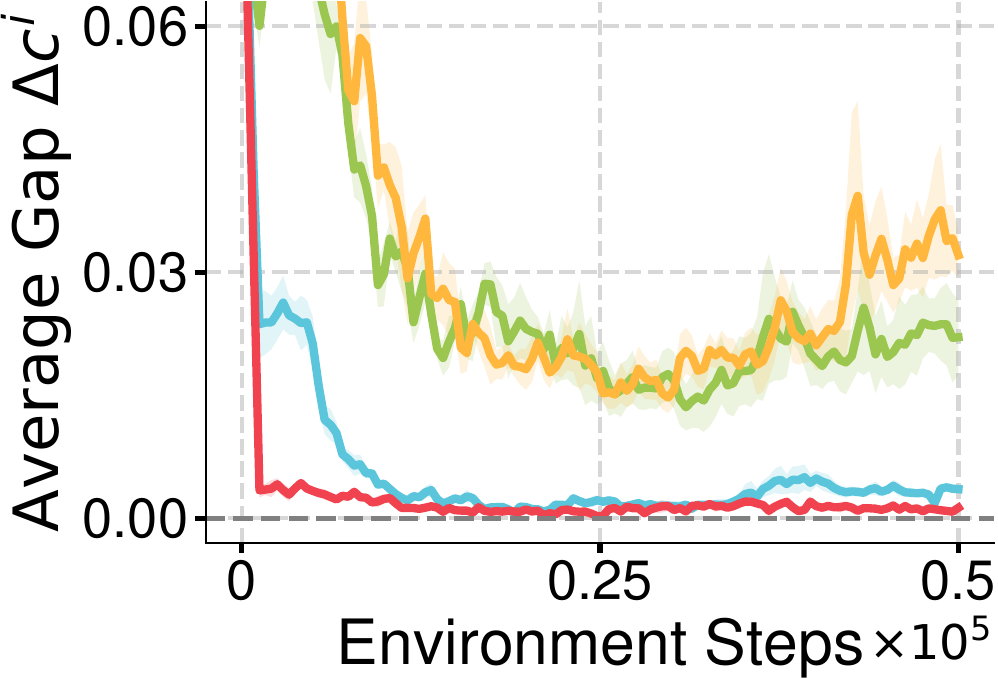}%
        \label{subfig:dt_diff_diff}
    }%
    \subfigure[Performance]{
        \includegraphics[width=0.48\linewidth]{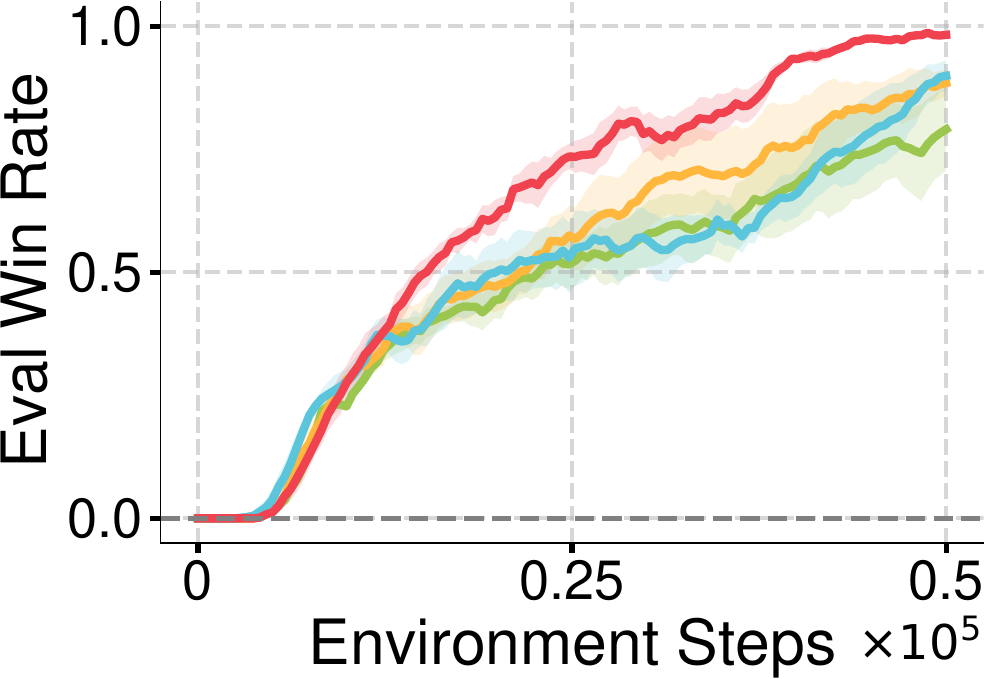}
        \label{subfig:dt_performance}
    }
    \caption{Off-policy correction comparisons in the SMAX-1s1z task.
         \subref{subfig:dt_diff1} Distance from true $\rho^i$,  \subref{subfig:dt_diff2} Distance from joint true $\boldsymbol{\rho}$,
         \subref{subfig:dt_diff_diff} Gap between (a) and (b),
         \subref{subfig:dt_performance} Final performance.
         Legend is shared across all plots.
         DT-ISR demonstrates the lowest gap $\Delta c^i$ and the highest performance.
    }
    \Description{Off-policy correction comparisons in the SMAX-1s1z task.
         \subref{subfig:dt_diff1} Distance from true $\rho^i$,  \subref{subfig:dt_diff2} Distance from joint true $\boldsymbol{\rho}$,
         \subref{subfig:dt_diff_diff} Gap between (a) and (b),
         \subref{subfig:dt_performance} Final performance.
         Legend is shared across all plots.
         DT-ISR demonstrates the lowest gap $\Delta c^i$ and the highest performance.}
    \label{fig:truncation_comparison}
\end{figure}

To motivate and evaluate our proposed method, we present key visualizations in Fig.~\ref{fig:truncation_comparison}. We choose the environment SMAX-\texttt{1s1z} with $0.5M$ training timestep, in which two agents with different dynamics cooperate to destroy enemies, thus accurate credit assignment can drive rapid learning. Here, we take \( \eta = 1.05 \) for \( c_t^{i,\text{DT}} \). The Fig.~\ref{subfig:dt_diff1}, \ref{subfig:dt_diff2} illustrate the average distance of each truncation weight from the true individual ISR \( \rho^i \) and joint ISR \( \boldsymbol{\rho} \), denoted as \( d(\rho^i, c^i) \) and \( d(\boldsymbol{\rho}, c^i) \), respectively. Fig.~\ref{subfig:dt_diff_diff} highlights the average gap between these two distances, which shows how the weight is balanced between the two true ISRs. The average gap is defined as 
\begin{equation}
\label{eq:delta_c}
    \Delta c^i := \mathbb{E}_{\tau}[|d(\boldsymbol{\rho}, c^i) - d(\rho^i, c^i)|],
\end{equation}
where $c^i$ is given off-policy correction term and $\tau$ is a trajectory sample. The smaller this gap, the better the truncation method is at approximating the true joint ISR while maintaining accurate individual ISR estimates. 
As shown in Fig.~
\ref{subfig:dt_diff_diff}, DT-ISR is effective in reducing the gap $\Delta c^i$ when compared to other methods. Nonetheless, further validation is required to confirm that the minimization of the gap correlates with stable learning. Fig.~\ref{subfig:dt_performance} compares the performance of DT against IT, ST, and without off-policy correction, demonstrating the superiority of DT in achieving higher win rates.

This analysis underscores the advantages of our DT-ISR, laying a strong foundation for further exploration of robust off-policy corrections in MARL.

\subsection{Overall Learning Structure and Algorithm}
\label{subsec:overall_algo}

We now summarize the overall training pipeline of GPAE with DT-ISR under CTDE. The procedure integrates the on- and off-policy estimators into a practical optimization loop. At each iteration, (i) trajectories are collected by the joint policy $\boldsymbol{\pi}=\prod_{i=1}^N \pi^i$, (ii) stored into a short replay buffer $R$, and (iii) gradient updates are performed by sampling mini-batches from $R$. Actors $\boldsymbol{\theta}$ are updated per agent, and a centralized critic head parameterizes $\EQi(s,a^{-i};\boldsymbol{\psi})$.

For each sampled trajectories, we compute GPAE by Eq.~\ref{eq:gpae_off}, and the value function is optimized using the following loss function:
\begin{equation}
\label{eq:value_loss1}
    L(\psi) = \sum_{\tau \in \mathcal{B}} (\overline{EQ}^i_{\psi} - \overline{EQ}_{targ}^i)^2,
\end{equation}
where $\mathcal{B}$ represents the replay buffer, and the target value $\overline{EQ}_{targ}^i$ is explicitly given   by 
\begin{equation}
\label{eq:target_value}
    \overline{EQ}_{targ}^i = \overline{EQ}^i_{\overline{\psi}} + \overline{\rho_t^i}\hat{A}_t^{i,GPAE},
\end{equation}

where $\overline{\rho_t^i} = \min(1, \rho_t^i)$ is a bounded individual ISR for stable learning and $\overline{\psi}$ denotes a parameter of target value network.
Actors are updated with PPO-style clipped objectives using $A^i_t$.

The complete procedure is summarized in Algorithm~\ref{alg:gpae}, which unifies our theoretical operator into an implementable CTDE algorithm with off-policy sample reuse.

\begin{algorithm}[h]
   \caption{GPAE}
   \label{alg:gpae}
\begin{algorithmic}
    \STATE Initialize shared policy parameter  $\boldsymbol{\theta}$, shared value parameter  $\boldsymbol{\psi}$ for $N$  agents
    \STATE Initialize trajectory buffer $R$ with batch reuse length $M$ 
    \FOR{each iteration $m$}
    \STATE Collect a set of trajectories $B_m$  from the environment by the joint policy $\boldsymbol{\pi}_{m}$ 
    \STATE Store  $B_m$  in the trajectory buffer $R$ 
    \STATE Compute  $\hat{A}_t^{i,\text{GPAE}}$  and  $\overline{EQ}_{targ}^i$  with \\Eq.~\ref{eq:gpae_off} and Eq.~\ref{eq:target_value} by using the samples stored in  $R$ 
    \FOR{each gradient step}
    \STATE Update  the actor parameter $\boldsymbol{\theta}$  with PPO actor loss
    \STATE Update  the critic parameter $\boldsymbol{\psi}$  with Eq.~\ref{eq:value_loss1}
    \ENDFOR
    \ENDFOR
\end{algorithmic}
\end{algorithm}

%% file: 5.experiments.tex
\section{Experimental Results}
\label{sec:experiment}
\subsection{Experimental Setup}
We evaluated the proposed method in two challenging MARL environments, \textbf{SMAX} and \textbf{MABrax}, implemented in JAX~\cite{jax}. These environments are representative of discrete and continuous action domains respectively, and serve as standard testbeds for evaluating GPAE’s credit assignment under partial observability. Figure~\ref{fig:envs} provides an overview of the two environments used in our evaluation.

\begin{figure}[ht]
\centering
\subfigure[\texttt{Halfcheetah-6x1} task of MABrax]{
    \includegraphics[width=0.53\linewidth]{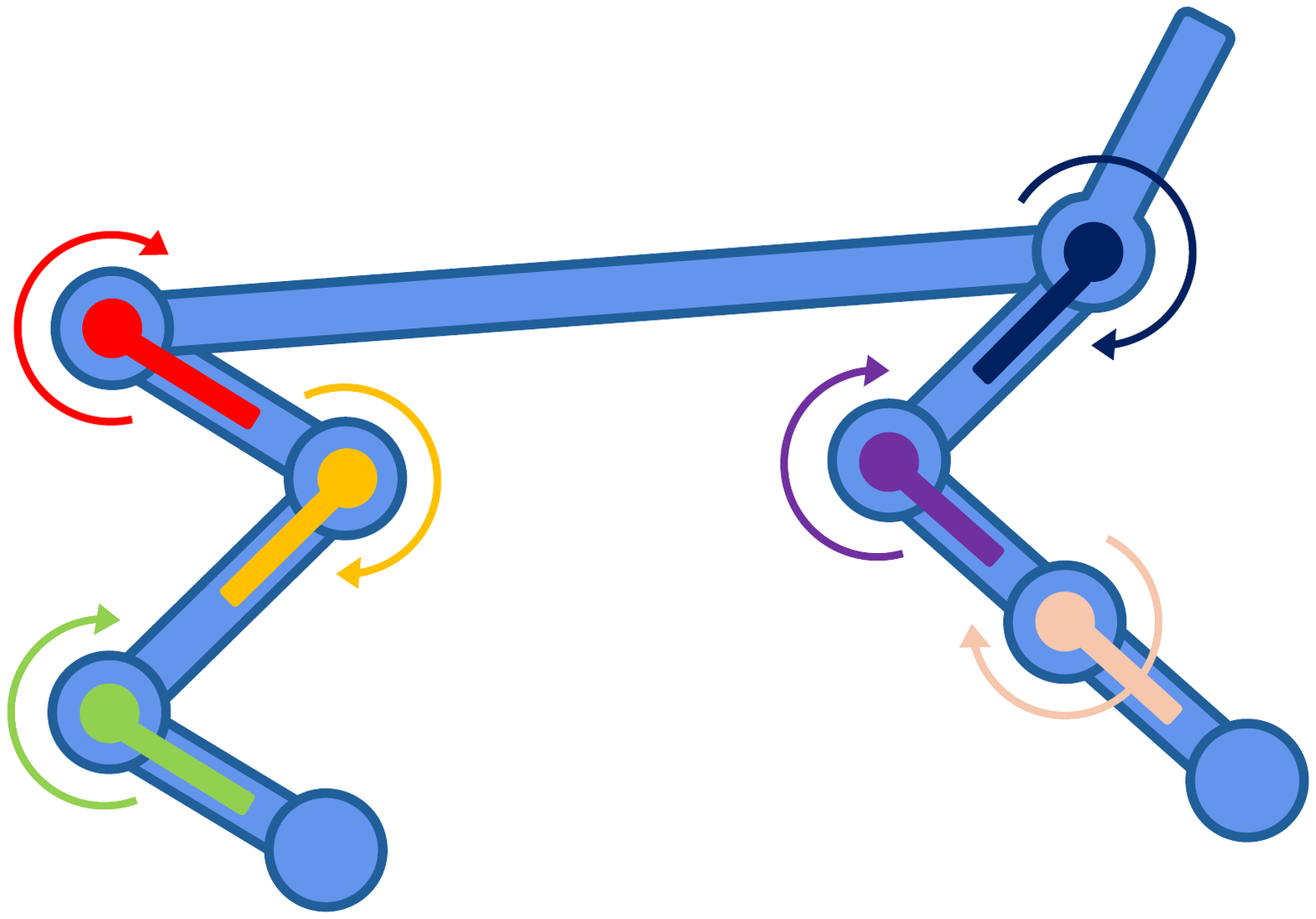}
    \label{fig:env_mabrax}
}%
\subfigure[\texttt{3s5z} task of SMAX]{
    \includegraphics[width=0.43\linewidth]{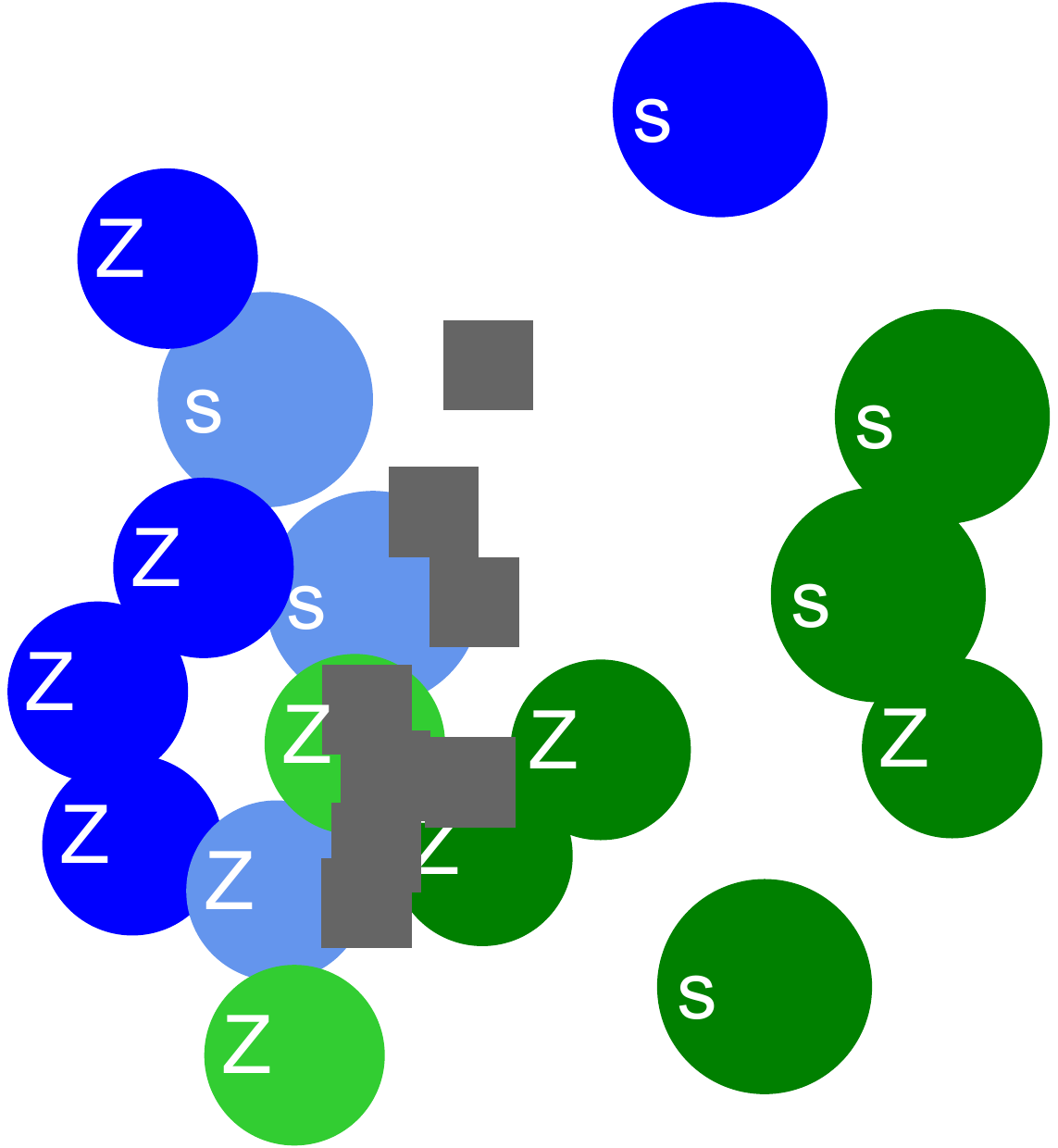}
    \label{fig:env_smax}
}
\caption{
Illustration of the two evaluation environments.}
\Description{Illustration of the two evaluation environments.}
\label{fig:envs}
\end{figure}

\textbf{MABrax} is a Multi-agent version of Brax~\cite{brax}, where each agent represents an individual joint of a robot tasked with coordinated movement, with continuous action spaces. MABrax focuses on multi-agent continuous control, with each agent corresponding to an individual joint of a robot. proper credit assignment based on centralized values is crucial to ensure synchronized actions among the joints. Limited observations further emphasize the need for precise credit assignment, as effective coordination is necessary to complete tasks successfully. Figure~\ref{fig:env_mabrax} illustrates the \texttt{halfcheetah-6x1} task, where six agents independently control the hip, thigh, and leg joints of a half-cheetah robot to generate coordinated forward locomotion.

\textbf{SMAX} is a JAX-based implementation of the StarCraft Multi-agent Challenge (SMAC, \citet{smac}). SMAX involves discrete action spaces where effective inter-agent cooperation is critical for task success. Agents engage in combat to achieve a common goal, with the win rate serving as the primary evaluation metric. Multi-agent credit assignment plays a pivotal role in enabling agents to coordinate under partial observability. Figure~\ref{fig:env_smax} includes the \texttt{3s5z} scenario, where three Stalker units and five Zealot units form a cooperative team against a mirrored adversarial team. The gray squares in the figure visualize attacks between units.

\noindent \textbf{Baselines.}
Our method is compared against baseline algorithms, focusing on its performance relative to other existing advantage estimating MAPG methods. Specifically, we consider representative baselines including MAPPO~\cite{mappo}, which utilizes GAE\cite{gae} as an advantage estimator. We also compare against DAE\cite{dae}, a variant of advantage estimation using reward estimation, and COMA~\cite{coma}, a classical counterfactual baseline for credit assignment. In addition, QMIX~\cite{QMIX} and VDN~\cite{vdn} are included as value decomposition approaches for discrete-action tasks. Note that QMIX and VDN are not applicable for continuous-action tasks. Furthermore, we provide an ablation study to assess the impact of the DT-ISR scheme and analyze the influence of hyperparameter $\eta$ on learning performance.

Including the environment, the implementation of MAPPO, VDN and QMIX is based on \citet{jaxmarl}. For GPAE, we provide our own implementation, where off-policy sample reuse is achieved through a short replay buffer. Since DAE has no existing public implementation, we implemented it from scratch. Similarly, COMA lacks a JAX-based implementation, so we implemented it closely following the original description.

\noindent \textbf{Evaluation Protocol.}
All methods were trained for 10M timesteps in each environment, and evaluated every 0.5M timesteps. For each task, we conducted 5 independent runs with different random seeds and report the mean and standard deviation across seeds. In SMAX, performance is measured by the win rate, while in MABrax it is measured by episodic return. The results reported in the main tables correspond to the final performance at the end of training. In addition, we provide aggregated learning curves that combine results across all tasks within each environment, where the shaded region indicates the standard error of the mean. 
Computationally, our algorithm increases wall-clock training time by at most 
6\% compared to GAE within same MAPPO across all SMAX and MABrax tasks. All experiments were conducted on a server equipped with NVIDIA TITAN Xp GPUs and Intel(R) Xeon(R) CPU E5-2620 v4 @ 2.10GHz.

\begin{table*}[ht]
    \renewcommand{\arraystretch}{1.1}
    \centering
    \caption{Performance comparison after learning $1 \times 10^7$ timesteps. The upper section shows the average win rate performances in SMAX tasks, with values presented as mean (standard deviation). The lower section shows the episode returns in MABrax tasks, with values presented as mean $\pm$ standard deviation. For each task, we conducted experiments with 5 random seeds. The bold numerals denote the highest performance for each task. Note that QMIX and VDN are not applicable to continuous tasks, and therefore, no results are available for MABrax tasks.}

    \begin{tabular}{c|c|ccccccc}
\toprule
Domain&Task & GPAE (off) & GPAE (on) & MAPPO & DAE & COMA & QMIX & VDN \\
\midrule
\multirow{6}{*}{SMAX} & 3s5z\_vs\_3s6z & \textbf{87.3}{\scriptsize$\pm$3.9} & 6.1{\scriptsize$\pm$2.0} & 2.6{\scriptsize$\pm$0.7} & 6.5{\scriptsize$\pm$2.0} & 0.0{\scriptsize$\pm$0.0} & 1.3{\scriptsize$\pm$0.8} & 0.5{\scriptsize$\pm$0.4} \\
&5m\_vs\_6m & \textbf{93.7}{\scriptsize$\pm$1.0} & 8.8{\scriptsize$\pm$0.9} & 3.1{\scriptsize$\pm$1.8} & 4.2{\scriptsize$\pm$2.5} & 0.3{\scriptsize$\pm$0.6} & 3.9{\scriptsize$\pm$0.6} & 3.2{\scriptsize$\pm$1.9} \\
&10m\_vs\_11m & \textbf{98.5}{\scriptsize$\pm$2.7} & 47.5{\scriptsize$\pm$2.8} & 36.4{\scriptsize$\pm$2.4} & 25.0{\scriptsize$\pm$8.2} & 1.3{\scriptsize$\pm$0.3} & 32.8{\scriptsize$\pm$0.9} & 22.6{\scriptsize$\pm$8.0} \\
&6h\_vs\_8z & 99.5{\scriptsize$\pm$0.1} & \textbf{99.8}{\scriptsize$\pm$0.2} & 99.0{\scriptsize$\pm$0.2} & 98.9{\scriptsize$\pm$0.1} & 80.4{\scriptsize$\pm$4.9} & 75.3{\scriptsize$\pm$5.3} & 96.1{\scriptsize$\pm$2.6} \\
&smacv2\_5\_units & \textbf{81.0}{\scriptsize$\pm$1.3} & 80.0{\scriptsize$\pm$1.7} & 75.2{\scriptsize$\pm$0.9} & 77.1{\scriptsize$\pm$1.2} & 54.0{\scriptsize$\pm$1.9} & 59.5{\scriptsize$\pm$3.3} & 39.4{\scriptsize$\pm$2.9} \\
&smacv2\_10\_units & \textbf{75.0}{\scriptsize$\pm$1.3} & 69.0{\scriptsize$\pm$1.7} & 63.5{\scriptsize$\pm$2.5} & 65.7{\scriptsize$\pm$0.6} & 47.4{\scriptsize$\pm$0.4} & 62.1{\scriptsize$\pm$1.6} & 29.5{\scriptsize$\pm$2.4} \\
\hline
\multirow{6}{*}{MABrax} &halfcheetah-6x1 & \textbf{3463}{\scriptsize$\pm$68} & 3138{\scriptsize$\pm$32} & 2965{\scriptsize$\pm$45} & 2983{\scriptsize$\pm$26} & 2017{\scriptsize$\pm$39} & N/A & N/A \\
&ant-8x1 & \textbf{3285}{\scriptsize$\pm$151} & 1557{\scriptsize$\pm$82} & 1247{\scriptsize$\pm$49} & 1279{\scriptsize$\pm$40} & 901{\scriptsize$\pm$96} & N/A & N/A \\
&ant-4x2 & \textbf{3574}{\scriptsize$\pm$217} & 1666{\scriptsize$\pm$32} & 1379{\scriptsize$\pm$124} & 1441{\scriptsize$\pm$50} & 950{\scriptsize$\pm$27} & N/A & N/A \\
&walker2d-6x1 & \textbf{912}{\scriptsize$\pm$18} & 697{\scriptsize$\pm$19} & 489{\scriptsize$\pm$77} & 591{\scriptsize$\pm$25} & 242{\scriptsize$\pm$64} & N/A & N/A \\
&hopper-3x1 & \textbf{1572}{\scriptsize$\pm$110} & 1356{\scriptsize$\pm$66} & 1088{\scriptsize$\pm$65} & 1050{\scriptsize$\pm$63} & 709{\scriptsize$\pm$203} & N/A & N/A \\
&humanoid-9|8 & \textbf{445}{\scriptsize$\pm$4} & 286{\scriptsize$\pm$4} & 258{\scriptsize$\pm$7} & 261{\scriptsize$\pm$6} & 214{\scriptsize$\pm$45} & N/A & N/A \\
\bottomrule
\end{tabular}

    \label{tab:performance}
\end{table*}
\subsection{Comparison with Baselines} 

Across SMAX tasks, GPAE consistently achieves superior performance over all baseline methods. A notable trend is that MAPPO-based advantage estimators already outperform value decomposition baselines such as QMIX and VDN, underscoring the effectiveness of policy-gradient approaches in cooperative MARL. Building on this, GPAE delivers an additional margin of improvement, demonstrating its ability to stabilize training with precise credit assignment and improve sample efficiency with off-policy sample reuse. Importantly, GPAE shows its benefits even when restricted to on-policy samples only, already surpassing MAPPO and DAE. This highlights that the improvement is not solely due to sample reuse but also stems from the design of our estimator. 

When extended with off-policy sample reuse, GPAE achieves further gains in both performance and sample efficiency. In Figures~\ref{fig:main_results}, which aggregate performance across all SMAX and MABrax tasks respectively, GPAE exhibits clearly steeper learning curves at early timesteps, indicating that it learns effective coordination policies with fewer samples. This advantage is particularly evident in the most challenging SMAX scenarios such as \texttt{3s5z\_vs\_3s6z} and \texttt{5m\_vs\_6m}, where credit assignment is nontrivial. While MAPPO and DAE degrade under these conditions, GPAE sustains more reliable coordination. By contrast, COMA shows clear weaknesses, with its one-step counterfactual structure failing to capture long-range dependencies, leading to consistently poor performances. QMIX and VDN demonstrate relatively steady but inferior results compared to MAPPO-based estimators, further highlighting that value decomposition is less suitable for these SMAX tasks.

\begin{figure}[ht]
    \centering
    \includegraphics[width=0.8\linewidth]{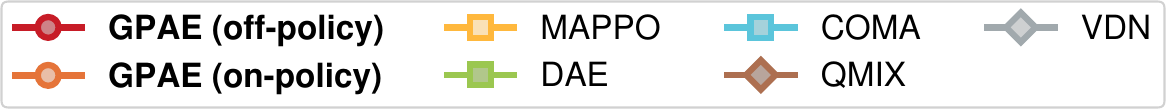}

    \subfigure[SMAX]{
        \includegraphics[width=0.49\linewidth]{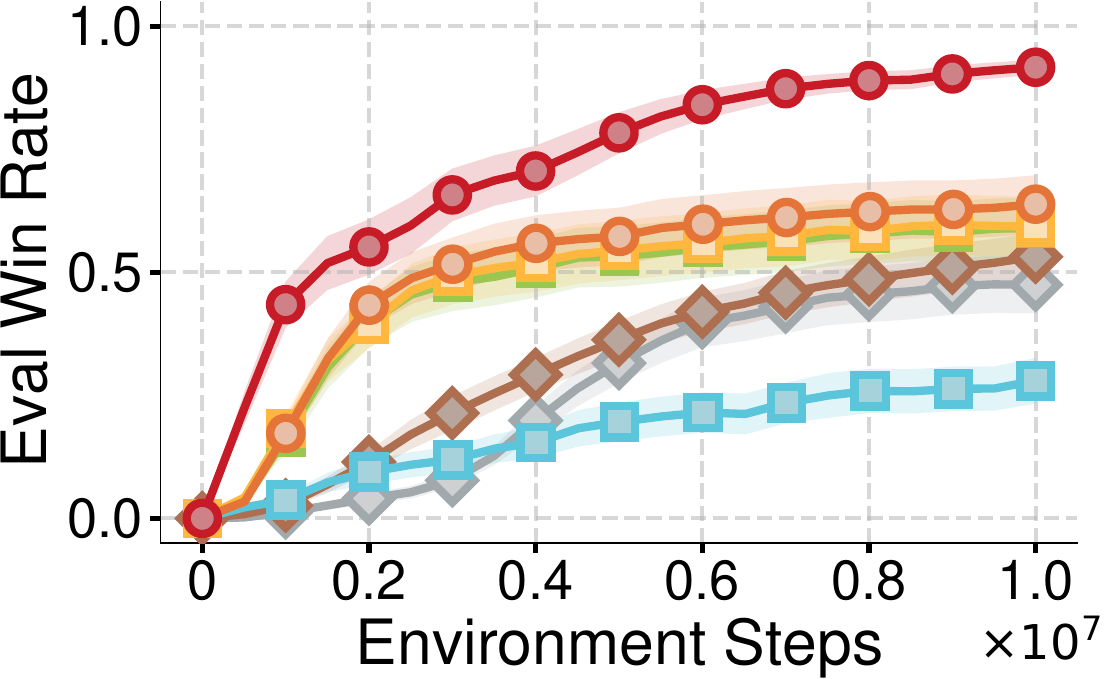}%
    }%
    \subfigure[MABrax]{
        \includegraphics[width=0.49\linewidth]{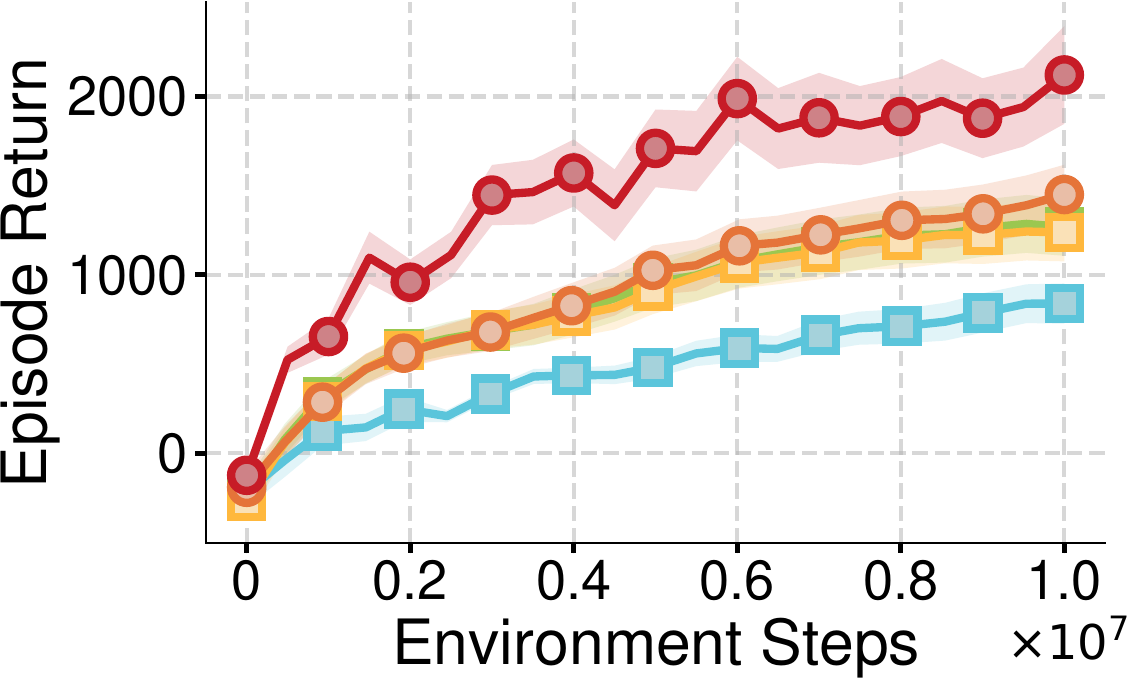}%
    }
    \caption{(a) and (b) represent the learning curve of the aggregated performance for all tasks on SMAX and MABrax, respectively.}
    \Description{(a) and (b) represent the learning curve of the aggregated performance for all tasks on SMAX and MABrax, respectively.}
    \label{fig:main_results}
\end{figure}

Turning to the continuous-control benchmarks in MABrax, the gap between GPAE and other baselines remains significant. Value decomposition methods such as QMIX and VDN are not applicable in continuous domains, leaving MAPG estimators and COMA as the primary points of comparison. Here, GPAE again surpasses MAPPO and DAE, which, while competitive, struggle to maintain stability across high-dimensional action spaces. COMA proves highly unstable in this setting, similar to results from discrete tasks, further confirming the limitations of traditional counterfactual baselines based on TD(0) estimation.

Taken together, these findings highlight that GPAE consistently outperforms the baselines across diverse domains. Even without off-policy augmentation, GPAE shows improvements over the existing methods, while with sample reuse it achieves unmatched sample efficiency and performance. In both discrete and continuous domains, GPAE establishes itself as a unified framework that not only improves asymptotic performance but also accelerates learning, thereby validating its effectiveness as a solution for multi-agent credit assignment.

\subsection{Ablation Study}

\noindent \textbf{DT-ISR Weight.}
We isolate the effect of the truncation scheme on \emph{GPAE} by fixing the advantage estimator and varying only the off-policy weighting: no correction, single truncation (ST), individual truncation (IT), and our double truncation (DT-ISR). All variants employ the same off-policy sample reuse protocol described in Section~\ref{sec:experiment}. We focus on two SMAX scenarios, \texttt{5m\_vs\_6m} and \texttt{3s5z}/\texttt{3s6z}. As summarized in Table~2, these tasks exhibit the largest performance gaps between GPAE and other advantage estimators. The ablation is designed to verify how success is realized on the most demanding settings when off-policy data is reused. In particular, the analysis asks whether stability can be retained while preserving per-agent credit assignment.

Table~\ref{tab:double_truncation} reports win rate performance after 10M timesteps. DT-ISR consistently outperforms other truncation schemes, achieving higher final performance. These results are consistent with the findings in Section~\ref{subsec:dt}, where DT-ISR improved off-policy learning stability across a broader set of tasks. The improvement here stems from the same mechanism: balancing stability and credit sensitivity by combining the strengths of joint and individual truncation.

\begin{table}[h]
\centering
    \centering
    \renewcommand{\arraystretch}{1.2}
    \caption{Ablation study on DT-ISR.}
    \begin{tabular}{l|c|c}
        \toprule
        method & \texttt{5m\_vs\_6m} & \texttt{3s5z\_vs\_3s6z}  \\
        \midrule
        \textbf{DT-ISR} & \textbf{93.7} & \textbf{87.3} \\
        ST-ISR & 44.4 & 80.8 \\
        IT-ISR & 58.6 & 83.7 \\
        No Correction & 34.5 & 74.9 \\
        \bottomrule
    \end{tabular}
    \label{tab:double_truncation}
\end{table}
\begin{table}[h]
    \centering
    \renewcommand{\arraystretch}{1.2}
    \caption{Impact of $\eta$ on DT-ISR weights.}
    \begin{tabular}{l|c|c|c|c}
        \toprule
        \multirow{2}{*}{$\eta$} & \multicolumn{2}{c|}{\texttt{5m\_vs\_6m}} & \multicolumn{2}{c}{\texttt{3s5z\_vs\_3s6z}} \\
        \cline{2-5}
        & win rate & $\Delta c^i$ & win rate & $\Delta c^i$ \\
        \midrule
        1.0 & \textbf{95.2} & -0.00154 & 84.5 & -0.00105 \\
        1.056 & 93.7 & -0.00409 & \textbf{87.3} & -0.0067 \\
        1.1 & 93.1 & -0.00553 & 83.2 & -0.00918 \\
        1.15 & 39.4 & -0.00561 & 80.3 & -0.00938 \\
        \bottomrule
    \end{tabular}
    \label{tab:eta_impact}
\end{table}

\noindent \textbf{Impact of $\eta$ and Trace Dynamics.} We analyze the impact of varying the double-truncation parameter $\eta$ on key metrics, including learning performance in SMAX, measured by win rate. We also compute \( \Delta c^i \) as defined in Eq.~\ref{eq:delta_c}, which measures the difference between the distance of truncation weights from the true joint ISR and individual ISR. The results are presented in Table~\ref{tab:eta_impact}, showing that DT-ISR is robust to changes in \( \eta \), with performance remaining stable across a range of values. Specifically, performance is consistently good for \( \eta \) values between 1.0 and 1.05.

The ablation suggests that the performance gain arises from applying off-policy weights that integrate per-agent $n$-step advantages, control variance, and preserve credit assignment. DT-ISR supplies this coupling with a single hyperparameter \(\eta\).

%% file: 6.relatedwork.tex
\section{Related Works}
\label{sec:relatedwork}

\noindent \textbf{Multi-Agent Policy Gradient.} 
\balance
Multi-agent policy gradient is an effective approach under CTDE because it can exploit the global state to estimate the centralized critic. 
MADDPG \citep{maddpg} uses  individual critic for each agent with global information. 
PS-TRPO \citep{pstrpo} proposes parameter sharing for actor-critic to reduce the network size.
COMA \citep{coma} estimates individual advantage with a counterfactual baseline, which enables explicit credit assignment.
DOP \citep{dop} and FACMAC \citep{facmac} propose off-policy policy gradient methods based on value decomposition.
IPPO \citep{ippo} and MAPPO \citep{mappo} show that PPO performs strongly in the SMAC environment.
HAPPO \citep{happo} proposed a sequential update structure of MAPPO, with individual actor networks.
\citet{ob} analyzed the variance of MAPG, and proposed a baseline that minimizes the variance \citep{ob}.

\noindent \textbf{Multi-Agent Credit Assignment Problem.} Credit assignment methods are commonly categorized into two approaches: \textit{implicit} and \textit{explicit} credit assignment. 
Implicit approaches such as VDN \cite{vdn}, QMIX~\citep{QMIX}, QTRAN~\citep{qtran}, and several subsequent studies \cite{qplex, ader} rely on value decomposition to implicitly infer individual contributions from the global Q-function.
Explicit methods aim to assign credit directly, often through reward shaping or advantage decomposition. LIIR~\citep{liir} learns intrinsic rewards based on individual state-action pairs. COMA~\citep{coma} introduces a counterfactual baseline for computing per-agent advantages, while DAE~\citep{dae} extends this idea by incorporating potential-based difference rewards into GAE~\citep{gae}. Closely related, \citet{castellini2022difference} propose a principled formulation of policy gradients that leverage marginal contribution signals.
Beyond these primary categories, methods such as MAAC~\citep{maac} and intention sharing~\citep{intention} employ attention mechanisms to enhance inter-agent coordination through structured representation learning.

\noindent \textbf{Off-Policy Value Estimation.} $n$-step  estimation can provide more accurate value estimation for policy gradient methods. TD($\lambda$) \citep{tdlambda} and GAE \citep{gae} has been used widely for on-policy methods.
In single-agent RL, $n$-step off-policy value estimation has been developed for improved sample efficiency. Retrace($\lambda$) \citep{retrace} and V-trace \citep{vtrace} propose  low-variance value estimation by using truncated importance sampling trace. DISC \citep{disc} proposed using GAE-V (i.e., GAE with V-trace) for off-policy advantage estimation. GePPO \citep{geppo} and \citet{offpolicyppo} analyzes the reuse of off-policy samples with PPO. For the multi-agent value-based method, SMIX($\lambda$) \citep{smix} extends QMIX by estimating the joint action value function based on  $n$-step return and expected SARSA \citep{sutton1}, where credit assignment is still implicit by value factorization.

%% file: 7.discussion.tex
\section{Conclusion}
\label{sec:discussion}
One of the most important challenges in MARL is addressing the multi-agent credit assignment problem. While various approaches have attempted to tackle this issue, explicit study and analysis of credit assignment within policy gradient-based methods have been relatively underexplored. Our work directly addresses this gap by introducing a novel framework that accurately estimates per-agent advantages, enabling more effective policy updates. By leveraging our new value operator and integrating advanced techniques like double-truncated importance sampling weight, we provide a principled approach to credit assignment that enhances both on-policy and off-policy learning. With its solid theoretical foundation, GPAE paves the way for the development of more efficient, robust, and scalable multi-agent systems.

%% file: appendix.tex
\section{Proofs}
\label{appendix:proof}

\subsection{Proof of Theorem \ref{thm:contraction_on}}
\label{proof:contraction_on}
We show that the operator $\mathcal{R}_{\text{on}}^i$ is a  $\gamma$-contraction under $L_\infty$-norm. With the function $\overline{EQ}^i = \mathbb{E}_{a^\sim \pi^i}[Q(s,a^i,\boldsymbol{a}^{-i})]$, we define the notation $\overline{EQ}^i_t := \mathbb{E}_{a^i_t\sim \pi^i}[Q(s_t,a_t^i,\boldsymbol{a}^{-i}_t)]$. Then, we have
\begin{align}
\mathcal{R}_{\text{on}}^i\overline{EQ}^i &= \overline{EQ}^i + \mathbb{E}_{a_0^i\sim\pi^i}\Bigg[\mathbb{E}_{\boldsymbol{\pi}}\Bigg[\sum_{t\geq 0} (\gamma\lambda)^t \left(r_t + \gamma  \overline{EQ}^i_{t+1} - \overline{EQ}^i_t\right)\nonumber 
\\
& \qquad \qquad\qquad\qquad\qquad\qquad\qquad\qquad \biggr| s_0, a_0^i,\boldsymbol{a}_0^{-i}\Bigg]\Bigg] \label{eq:contraction_on_1}
\\
&= \mathbb{E}_{a_0^i\sim\pi^i}\Bigg[\mathbb{E}_{\boldsymbol{\pi}}\Bigg[\sum_{t\geq 0} (\gamma\lambda)^t \left(r_t + \gamma(1-\lambda)  \overline{EQ}^i_{t+1} \right)\nonumber 
\\
& \qquad \qquad\qquad\qquad\qquad\qquad\qquad\qquad \biggr| s_0, a_0^i,\boldsymbol{a}_0^{-i}\Bigg]\Bigg], \label{eq:contraction_on_2}
\end{align}
where $(s_0,\boldsymbol{a}_0)=(s,\boldsymbol{a})$. Eq. ~\eqref{eq:contraction_on_2} is obtained by rearranging the time index for $- \overline{EQ}^i_t$ term inside the double expectation of the RHS of \eqref{eq:contraction_on_1}. Now, consider two functions $\overline{EQ}^{i,1}$ and 
$\overline{EQ}^{i,2}$ both with the argument $(s,\boldsymbol{a}^{-i})$ and the difference between the two operators' outputs $\mathcal{R}_{\text{on}}^i\overline{EQ}^{i,1}$ and $\mathcal{R}_{\text{on}}^i\overline{EQ}^{i,2}$: With $(s_0,\boldsymbol{a}_0)=(s,\boldsymbol{a})$,
\begin{align}
    &\mathcal{R}_{\text{on}}^i\overline{EQ}^{i,1} - \mathcal{R}_{\text{on}}^i\overline{EQ}^{i,2} \nonumber\\
    &= \mathbb{E}_{a_0^i\sim\pi^i}\Bigg[\mathbb{E}_{\boldsymbol{\pi}}\Bigg[\sum_{t\geq 0} (\gamma\lambda)^t \Big(\gamma(1-\lambda)  (\overline{EQ}^{i,1}(s_{t+1},\boldsymbol{a}_{t+1}^{-i}) - \nonumber\\
    & \qquad\qquad\qquad\qquad\qquad\qquad\overline{EQ}^{i,2}(s_{t+1},\boldsymbol{a}_{t+1}^{-i})) \Big)\biggr| s_0, a_0^i,\boldsymbol{a}_0^{-i}\Bigg]\Bigg]
    \\
    &= \mathbb{E}_{a_0^i\sim\pi^i}\left[\mathbb{E}_{\boldsymbol{\pi}}\left[\sum_{t\geq 1} (\gamma\lambda)^t (\lambda^{-1} -1)  \Delta\overline{EQ}^i(s_{t},\boldsymbol{a}_{t}^{-i})\biggr| s_0, a_0^i,\boldsymbol{a}_0^{-i}\right]\right],
\end{align}
where $\Delta\overline{EQ}^i := \overline{EQ}^{i,1} - \overline{EQ}^{i,2}$. Now consider the $L_\infty$-norm of $\mathcal{R}^i\overline{EQ}^{i,1} - \mathcal{R}_{\text{on}}^i\overline{EQ}^{i,2}$:
\begin{align}
&\biggr\|\mathcal{R}_{\text{on}}^i\overline{EQ}^{i,1} - \mathcal{R}_{\text{on}}^i\overline{EQ}^{i,2}\biggr\|_\infty \nonumber\\
&=\biggr\| \mathbb{E}_{a_0^i\sim\pi^i}\left[\mathbb{E}_{\boldsymbol{\pi}}\left[\sum_{t\geq 1} (\gamma\lambda)^{t} (\lambda^{-1} -1)\Delta \overline{EQ}^i(s_t,\boldsymbol{a}^{-i}_t)\right]\right]\biggr\|_\infty \label{eq:thm41_3m1on} \\
&\le\biggr\| \mathbb{E}_{a_0^i\sim\pi^i}\left[\mathbb{E}_{\boldsymbol{\pi}}\left[\sum_{t\geq 1} (\gamma\lambda)^{t} (\lambda^{-1} -1) \sup_{s,\boldsymbol{a}^{-i}}\Delta \overline{EQ}^i(s,\boldsymbol{a}^{-i})\right]\right]\biggr\|_\infty \label{eq:thm41_3m1onxxx} \\
&\le \biggr| \mathbb{E}_{a_0^i\sim\mu^i}\left[\mathbb{E}_{\boldsymbol{\pi}}\left[\sum_{t\geq 1} (\gamma\lambda)^{t} (\lambda^{-1} -1)\right]\right]\biggr| \cdot  \biggr\|\Delta \overline{EQ}^i\biggr\|_\infty \label{eq:thm41_3on} \\
&=\sum_{t\geq 1}(\gamma\lambda)^{t} (\lambda^{-1} -1)\biggr\|\Delta \overline{EQ}^i\biggr\|_\infty \\
&= \gamma\lambda\frac{(\lambda^{-1} -1)}{1-\gamma\lambda} \biggr\|\Delta \overline{EQ}^i\biggr\|_\infty \\
&= \gamma\frac{(1-\lambda)}{1-\gamma\lambda} \biggr\|\Delta \overline{EQ}^i\biggr\|_\infty 
\end{align}
\begin{flalign}
\hspace{0pt} \le \gamma \biggr\|\Delta \overline{EQ}^i\biggr\|_\infty &&
\end{flalign}
since $\lambda \in (0,1]$, thus $\mathcal{R}_{\text{on}}^i$ is a $\gamma$-contraction.

With $\lambda=1$, the unique fixed point is given by $\mathbb{E}_{a^i\sim\pi^i}[Q^{\boldsymbol{\pi}}(s,a^i,$

\noindent $\boldsymbol{a}^{-i})]$. We can show this as follows. The inner conditional expectation of the second term in the RHS of \eqref{eq:contraction_on_1} is given by
\begin{align}
&=\mathbb{E}_{\mathbf{\tau}\sim\boldsymbol{\pi}}\left[\sum_{t\geq 0} \gamma^t  \left(r(s_t,\boldsymbol{a}_t) + \gamma  \overline{EQ}^i_{t+1} - \overline{EQ}^i_t\right)\biggr| s_0=s, \boldsymbol{a}_0=\boldsymbol{a}\right]\nonumber\\
&= \mathbb{E}_{\mathbf{\tau}\sim\boldsymbol{\pi}} \biggr[ ~~~\gamma^0 (r(s_0,a_0) +\gamma \overline{EQ}^i_1 - \overline{EQ}^i_0 ) \nonumber\\
    &    \hspace{3em}~~~+ \gamma^1 (r(s_1,a_1) +\gamma \overline{EQ}^i_2 - \overline{EQ}^i_1 )                \nonumber\\
    &    \hspace{3em}~~~+ \gamma^2 (r(s_2,a_2) +\gamma \overline{EQ}^i_3 - \overline{EQ}^i_2 )  \nonumber\\
    &   \hspace{3em}~~~+ ~~~~~~~~~~~~~~~~~~~\cdots  \hspace{13em}\biggr\rvert s_0=s, \boldsymbol{a}_0=\boldsymbol{a}\biggr]\nonumber
\\ &\stackrel{(a)}{=} \mathbb{E}_{\mathbf{\tau}\sim\boldsymbol{\pi}}\biggr[\sum_{t=0}^\infty \gamma^t r(s_t,\boldsymbol{a}_t)\biggr\rvert s_0=s, \boldsymbol{a}_0=\boldsymbol{a}\biggr] - \mathbb{E}_{\mathbf{\tau}\sim\boldsymbol{\pi}}\biggr[\overline{EQ}^i_0 \biggr\rvert s_0=s,\boldsymbol{a}_0=\boldsymbol{a}\biggr]\nonumber\\
&= Q^{\boldsymbol{\pi}}(s, \boldsymbol{a}) -\overline{EQ}^i_0|_{(s_0,\boldsymbol{a}^{-i}_0)=(s,\boldsymbol{a}^{-i})}\nonumber\\
&=Q^{\boldsymbol{\pi}}(s, \boldsymbol{a}) -\overline{EQ}^i,\label{eq:proofthm42_20}
\end{align}
where $\mathbf{\tau}=(s_1,\boldsymbol{a}_1,s_2,\boldsymbol{a}_2,\cdots)$ denotes the state-action trajectory.
Then,  the second term in the RHS of \eqref{eq:Ri_on} is expressed as
\begin{equation} \label{eq:proofthm42_21}
\mathbb{E}_{a^i\sim \pi^i}[Q^{\boldsymbol{\pi}}(s, \boldsymbol{a}))] -\mathbb{E}_{a^i\sim\pi^i}[Q(s,a^i,\boldsymbol{a}^{-i})].
\end{equation}
By adding the first term $\mathbb{E}_{a^i\sim\pi^i}[Q(s, a^i,\boldsymbol{a}^{-i})]$ in the RHS of \eqref{eq:Ri_on} to  \eqref{eq:proofthm42_21}, we have
\[
\mathcal{R}_{\text{on}}^i \mathbb{E}_{a^i\sim\pi^i}[Q(s,a^i,\boldsymbol{a}^{-i})] = \mathbb{E}_{a^i\sim \pi^i}[Q^{\boldsymbol{\pi}}(s, \boldsymbol{a})] 
\]
for any $Q$. In particular, we choose $Q=Q^{\boldsymbol{\pi}}$:
\begin{align*}
    \mathcal{R}_{\text{on}}^i \mathbb{E}_{a^i\sim\pi^i}[Q^{\boldsymbol{\pi}}(s,a^i,\boldsymbol{a}^{-i})] &= \mathbb{E}_{a^i\sim \pi^i}[Q^{\boldsymbol{\pi}}(s, \boldsymbol{a})]\\
    &=\mathbb{E}_{a^i\sim \pi^i}[Q^{\boldsymbol{\pi}}(s, a^i, \boldsymbol{a}^{-i})].
\end{align*}

Hence, $\mathbb{E}_{a^i\sim\pi^i}[Q^{\boldsymbol{\pi}}(s,a^i,\boldsymbol{a}^{-i})]$ is the fixed point of $\mathcal{R}_{\text{on}}^i$.
\qed

\subsection{Proof of Theorem \ref{thm:invariance_on}}
\label{proof:invariance_on}

\proof 
It is already shown that in the case with $\lambda=1$, the inner conditional expectation reduces to $Q^{\boldsymbol{\pi}}(s_0, \boldsymbol{a}_0) -\mathbb{E}_{a^i\sim\pi^i}[Q(s_0,a_0^i,\boldsymbol{a}_0^{-i})]$ in eq. \eqref{eq:proofthm42_20}. 

The invariance of policy gradient of \eqref{eq:proofthm42_20} can be shown in a similar way to the proof in \citep{coma}. That is, the policy gradient with respect to policy parameter $\theta$ with the advantage  \eqref{eq:proofthm42_20} becomes the true policy gradient $g = \mathbb{E}_{\boldsymbol{\pi}} \left[ \sum_i (\nabla_{\theta} \log \pi^i(a^i|o^i)) Q^{\boldsymbol{\pi}}(s,\boldsymbol{a}) \right]$ if the offset term $g_o = \mathbb{E}_{\boldsymbol{\pi}} \left[ \sum_i (\nabla_{\theta} \log \pi^i(a^i|o^i))\overline{EQ}^i) \right]=0$. This condition is valid under the decentralized execution assumption of $\boldsymbol{\pi}=\prod_{i=1}^N \pi^i(a^i|o^i)$ as shown below \citep{coma}:
\begin{align*}
    g_o &=  \mathbb{E}_{\boldsymbol{\pi}} \left[ \sum_i \nabla_{\theta} \log \pi^i(a^i|o^i) \overline{EQ}^i \right]\\ 
    &= \sum_i \sum_s d^{\boldsymbol{\pi}}(s) \sum_{\boldsymbol{a}}  \boldsymbol{\pi}(\boldsymbol{a}|\boldsymbol{o})   \nabla_{\theta} \log \pi^i(a^i|o^i)\overline{EQ}^i \\
    &= \sum_i \sum_s d^{\boldsymbol{\pi}}(s) \sum_{\boldsymbol{a}^{-i}} \sum_{a^{i}} \boldsymbol{\pi}(\boldsymbol{a}^{-i}|\boldsymbol{o}^{-i}) \pi^i(a^i|o^i)  \nabla_{\theta} \log \pi^i(a^i|o^i)\overline{EQ}^i \\
    &= \sum_i  \sum_s d^{\boldsymbol{\pi}}(s) \sum_{\boldsymbol{a}^{-i}} \boldsymbol{\pi}(\boldsymbol{a}^{-i}|\boldsymbol{o}^{-i})  \sum_{a^{i}} \pi^i(a^i|o^i)\nabla_{\theta} \log \pi^i(a^i|o^i)\overline{EQ}^i \\
    &= \sum_i\sum_s d^{\boldsymbol{\pi}}(s)  \sum_{\boldsymbol{a}^{-i}} \boldsymbol{\pi}(\boldsymbol{a}^{-i}|\boldsymbol{o}^{-i})  \sum_{a^{i}} \nabla_{\theta} \pi^i(a^i|o^i)\overline{EQ}^i \\
    &= \sum_i\sum_s d^{\boldsymbol{\pi}}(s)  \sum_{\boldsymbol{a}^{-i}} \boldsymbol{\pi}(\boldsymbol{a}^{-i}|\boldsymbol{o}^{-i})\overline{EQ}^i\nabla_{\theta}\sum_{a^{i}}  \pi^i(a^i|o^i) \\
    &= \sum_i\sum_s d^{\boldsymbol{\pi}}(s)  \sum_{\boldsymbol{a}^{-i}} \boldsymbol{\pi}(\boldsymbol{a}^{-i}|\boldsymbol{o}^{-i})\overline{EQ}^i\nabla_{\theta}1 \\
    &= 0. \qed
\end{align*}
\allowdisplaybreaks
\subsection{Proof of Theorem \ref{thm:contraction_off}}
\label{proof:contraction_off}
\textbf{Part I: Showing $\gamma$-contraction} 

We show that the operator $\mathcal{R}^i$ is a $\gamma$-contraction under  $L_\infty$-norm:
\begin{align}
&\mathcal{R}^i\overline{EQ}^i =\overline{EQ}^i_0 + \mathbb{E}_{a_0^i\sim\mu^i}\Bigg[\rho_0^i\mathbb{E}_{\boldsymbol{\mu}}\Bigg[\sum_{t\geq 0} \gamma^t \big(\prod_{j=1}^{t}c^{i}_j\big) \left(r_t + \gamma \overline{EQ}^i_{t+1} - \overline{EQ}^i_t\right) \nonumber
\\
&\qquad\qquad\qquad\qquad\qquad\qquad\qquad\qquad\biggr| s_0, a_0^i,\boldsymbol{a}_0^{-i}\Bigg]\Bigg] \label{eq:proof41_1}\\
&= (1 - \mathbb{E}_{a_0^i\sim\mu^i}[\rho_0^i])\overline{EQ}^i_0 + \mathbb{E}_{a_0^i\sim\mu^i}\Bigg[\rho_0^i\mathbb{E}_{\boldsymbol{\mu}}\Bigg[\sum_{t\geq 0} \gamma^t \big(\prod_{j=1}^{t}c^{i}_j\big) \\
&\qquad\qquad\qquad\qquad \cdot \left(r_t + \gamma  (1-c_{t+1}^i)\overline{EQ}^i_{t+1}\right)\biggr| s_0, a_0^i,\boldsymbol{a}_0^{-i}\Bigg]\Bigg], \label{eq:proof41_2}
\end{align}
where $(s_0,\boldsymbol{a}_0)=(s,\boldsymbol{a})$. Eq.~\eqref{eq:proof41_2} is obtained by rearranging the time index for $-\overline{EQ}^i_t$ term inside the double expectation of the RHS of \eqref{eq:proof41_1}.

Now, consider $\overline{EQ}^{i,1}$ and $\overline{EQ}^{i,2}$, in the same manner as \ref{proof:contraction_on}:
\begin{align}
\mathcal{R}^i\overline{EQ}^{i,1}-\mathcal{R}^i\overline{EQ}^{i,2}
&= (1 - \mathbb{E}_{a_0^i\sim\mu^i}[\rho_0^i])\Bigg(\overline{EQ}^{i,1}(s_{0},\boldsymbol{a}_0^{-i}) \nonumber \\
&\qquad\qquad\qquad\qquad\qquad- \overline{EQ}^{i,2}(s_0,\boldsymbol{a}_0^{-i})\Bigg) \nonumber \\
&\quad + \mathbb{E}_{a_0^i\sim\mu^i}\Bigg[\rho_0^i \, \mathbb{E}_{\boldsymbol{\mu}}\Bigg[
\sum_{t\geq 0} \gamma^{t+1} \left(\prod_{j=1}^{t}c^{i}_j\right)(1-c_{t+1}^i) \nonumber \\
&\qquad\cdot \left(\overline{EQ}^{i,1}(s_{t+1},\boldsymbol{a}_{t+1}^{-i}) - \overline{EQ}^{i,2}(s_{t+1},\boldsymbol{a}_{t+1}^{-i})\right)
\Bigg]\Bigg] \nonumber \\
&= (1 - \mathbb{E}_{a_0^i\sim\mu^i}[\rho_0^i])\Delta \overline{EQ}^i_0 \nonumber \\
& + \mathbb{E}_{a_0^i\sim\mu^i}\Bigg[\rho_0^i \, \mathbb{E}_{\boldsymbol{\mu}}\Bigg[
\sum_{t\geq 1} \gamma^{t} \left(\prod_{j=1}^{t-1}c^{i}_j\right)(1-c_{t}^i) \, \Delta \overline{EQ}^i_t
\Bigg]\Bigg]
\label{eq:proofprop1RiDiff}
\end{align}
where $\Delta \overline{EQ}^i := \overline{EQ}^{i,1} - \overline{EQ}^{i,2}$.  
Now consider the $L_\infty$-norm of $\mathcal{R}^i\overline{EQ}^i_1 - \mathcal{R}_{\text{on}}^i\overline{EQ}^i_2$:
\begin{align}
\biggr\|\mathcal{R}^i&\overline{EQ}^{i,1}(s,\boldsymbol{a}^{-i}) - \mathcal{R}^i\overline{EQ}^{i,2}(s,\boldsymbol{a}^{-i})\biggr\|_\infty \\
&=\biggr\| (1 - \mathbb{E}_{a_0^i\sim\mu^i}[\rho_0^i])\Delta \overline{EQ}^i_0 + \mathbb{E}_{a_0^i\sim\mu^i}\Bigg[\rho_0^i\mathbb{E}_{\boldsymbol{\mu}}\Bigg[\sum_{t\geq 1} \gamma^{t} \big(\prod_{j=1}^{t-1}c^{i}_j\big)\nonumber\\
&\qquad\qquad\qquad\qquad\qquad\qquad\qquad\cdot(1-c_{t}^i)\Delta \overline{EQ}^i_t\Bigg]\Bigg]\biggr\| _\infty\label{eq:thm41_3m1} \\
&\le \Bigg| (1 - \mathbb{E}_{a_0^i\sim\mu^i}[\rho_0^i])  \nonumber \\
&\qquad +\mathbb{E}_{a_0^i\sim\mu^i}\left[\rho_0^i\mathbb{E}_{\boldsymbol{\mu}}\left[\sum_{t\geq 1} \gamma^{t} \big(\prod_{j=1}^{t-1}c^{i}_j\big)(1-c_{t}^i)\right]\right]\Bigg| \cdot  \biggr\|\Delta \overline{EQ}^i\biggr\|_\infty, \label{eq:thm41_3}
\end{align}
where (\ref{eq:thm41_3}) is obtained by the definition of $L_\infty$-norm, the triangle inequality and the fact that the terms $(1 - \mathbb{E}_{a_0^i\sim\mu^i}[\rho_0^i])$ and $\gamma^{t} \big(\prod_{j=1}^{t-1}c^{i}_j\big)(1-c_{t}^i)$ in front of $\Delta \overline{EQ}^i$ in \eqref{eq:thm41_3m1} are all nonnegative. (Note that the trace coefficient $c_t^i \in [0,1]$.)

Now, consider the scaling term in front of $\biggr\|\Delta \overline{EQ}^i\biggr\|_\infty$ in \eqref{eq:thm41_3}. First, the second term in the front scaling term can be expressed as follows:
\begin{align}
&\mathbb{E}_{a_0^i\sim\mu^i}\left[\rho_0^i\mathbb{E}_{\boldsymbol{\mu}}\left[\sum_{t\geq 1} \gamma^{t} \big(\prod_{j=1}^{t-1}c^{i}_j\big)(1-c_{t}^i)\right]\right]  \nonumber\\
&= \mathbb{E}_{a_0^i\sim\mu^i}\left[\rho_0^i\mathbb{E}_{\boldsymbol{\mu}}\left[\sum_{t\geq 1} \gamma^{t} \big(\prod_{j=1}^{t-1}c^{i}_j\big)\right]\right] \nonumber\\
&\qquad\qquad\qquad\qquad\qquad\qquad- \mathbb{E}_{a_0^i\sim\mu^i}\left[\rho_0^i\mathbb{E}_{\boldsymbol{\mu}}\left[\sum_{t\geq 1} \gamma^{t} \big(\prod_{j=1}^{t-1}c^{i}_j\big)c_t^i\right]\right] \nonumber\\
&= \mathbb{E}_{a_0^i\sim\mu^i}\left[\rho_0^i\mathbb{E}_{\boldsymbol{\mu}}\left[\sum_{t\geq 1} \gamma^{t} \big(\prod_{j=1}^{t-1}c^{i}_j\big)\right]\right] \nonumber\\
&\qquad\qquad\qquad\qquad\qquad\qquad- \mathbb{E}_{a_0^i\sim\mu^i}\left[\rho_0^i\mathbb{E}_{\boldsymbol{\mu}}\left[\sum_{t\geq 1} \gamma^{t} \big(\prod_{j=1}^{t}c^{i}_j\big)\right]\right] \nonumber\\
&= \mathbb{E}_{a_0^i\sim\mu^i}\left[\rho_0^i\mathbb{E}_{\boldsymbol{\mu}}\left[\sum_{t\geq 1} \gamma^{t} \big(\prod_{j=1}^{t-1}c^{i}_j\big)\right]\right] \nonumber \\
& \qquad\qquad\qquad\qquad\qquad- \gamma^{-1}\mathbb{E}_{a_0^i\sim\mu^i}\left[\rho_0^i\mathbb{E}_{\boldsymbol{\mu}}\left[\sum_{t\geq 1} \gamma^{t+1} \big(\prod_{j=1}^{t}c^{i}_j\big)\right]\right] \nonumber\\
&= \mathbb{E}_{a_0^i\sim\mu^i}\left[\rho_0^i\mathbb{E}_{\boldsymbol{\mu}}\left[\sum_{t\geq 1} \gamma^{t} \big(\prod_{j=1}^{t-1}c^{i}_j\big)\right]\right] \nonumber\\
&\qquad\qquad\qquad- \gamma^{-1} \mathbb{E}_{a_0^i\sim\mu^i}\left[\rho_0^i\mathbb{E}_{\boldsymbol{\mu}}\left[-\gamma+\gamma+\sum_{t\geq 1} \gamma^{t+1} \big(\prod_{j=1}^{t}c^{i}_j\big)\right]\right] \nonumber\\
&= \mathbb{E}_{a_0^i\sim\mu^i}\left[\rho_0^i\mathbb{E}_{\boldsymbol{\mu}}\left[\sum_{t\geq 1} \gamma^{t} \big(\prod_{j=1}^{t-1}c^{i}_j\big)\right]\right] \nonumber\\
&\qquad\qquad\qquad- \gamma^{-1} \mathbb{E}_{a_0^i\sim\mu^i}\left[\rho_0^i\mathbb{E}_{\boldsymbol{\mu}}\left[-\gamma+\gamma+\sum_{t'\geq 2} \gamma^{t'} \big(\prod_{j=1}^{t'-1}c^{i}_j\big)\right]\right] \label{eq:thm41_6}\\
&= \mathbb{E}_{a_0^i\sim\mu^i}\left[\rho_0^i\mathbb{E}_{\boldsymbol{\mu}}\left[\sum_{t\geq 1} \gamma^{t} \big(\prod_{j=1}^{t-1}c^{i}_j\big)\right]\right] \nonumber\\
&\qquad\qquad\qquad\qquad- \gamma^{-1} \mathbb{E}_{a_0^i\sim\mu^i}\left[\rho_0^i\mathbb{E}_{\boldsymbol{\mu}}\left[-\gamma+\sum_{t\geq 1} \gamma^{t} \big(\prod_{j=1}^{t-1}c^{i}_j\big)\right]\right] \label{eq:thm41_7}\\
&= (1-\gamma^{-1})\mathbb{E}_{a_0^i\sim\mu^i}\left[\rho_0^i \mathbb{E}_{\boldsymbol{\mu}}\left[\sum_{t\geq 1} \gamma^{t} \big(\prod_{j=1}^{t-1}c^{i}_j\big)\right]\right] + \mathbb{E}_{a_0^i\sim\mu^i}\left[\rho_0^i\right]. \label{eq:thm41_8}
\end{align}

Here, \eqref{eq:thm41_6} is  obtained by defining the index $t':=t+1$. \eqref{eq:thm41_7} is valid because $\prod_{j=1}^0 c_j^i=1$ by definition. \eqref{eq:thm41_8} is obtained by factoring the common term out.  

Summing the first term $(1 - \mathbb{E}_{a_0^i\sim\mu^i}[\rho_0^i])$ and the second term expressed as \eqref{eq:thm41_8} in the front scaling term  of $\biggr\|\Delta \overline{EQ}^i\biggr\|_\infty$ in \eqref{eq:thm41_3}, we have the overall  scaling term, given by
\begin{align}
&(1 - \mathbb{E}_{a_0^i\sim\mu^i}[\rho_0^i])+
(1-\gamma^{-1})\mathbb{E}_{a_0^i\sim\mu^i}\left[\rho_0^i \mathbb{E}_{\boldsymbol{\mu}}\left[\sum_{t\geq 1} \gamma^{t} \big(\prod_{j=1}^{t-1}c^{i}_j\big)\right]\right] \nonumber\\
&\qquad\qquad\qquad\qquad\qquad\qquad\qquad\qquad\qquad + \mathbb{E}_{a_0^i\sim\mu^i}\left[\rho_0^i\right]\nonumber\\
&=1 +
(1-\gamma^{-1})\mathbb{E}_{a_0^i\sim\mu^i}\left[\rho_0^i \mathbb{E}_{\boldsymbol{\mu}}\left[\sum_{t\geq 1} \gamma^{t} \big(\prod_{j=1}^{t-1}c^{i}_j\big)\right]\right] \nonumber\\
&=1 +
(1-\gamma^{-1})\mathbb{E}_{a_0^i\sim\mu^i}\left[\rho_0^i \mathbb{E}_{\boldsymbol{\mu}}\left[\gamma \underbrace{\prod_{j=1}^0c_j^i}_{=1} \right.\right. \nonumber  \\
& \qquad\qquad\qquad\qquad\qquad \left.\left.+ \underbrace{\gamma^2 \prod_{j=1}^1 c_j^i+ \gamma^3 \prod_{j=1}^2 c_j^i+\cdots}_{\ge 0 }\right]\right] \label{eq:proofthm41_10m1}\\
&\le 1 + (1-\gamma^{-1})\gamma \mathbb{E}_{a_0^i\sim\mu^i}[\rho_0^i]\label{eq:proofthm41_10}\\
&= 1 + (1-\gamma^{-1})\gamma \label{eq:proofthm41_11} \\
&=\gamma. \nonumber
\end{align}
Here, \eqref{eq:proofthm41_10} is valid because the term inside the expectation $\mathbb{E}_{\boldsymbol{\mu}}[\cdot]$ in \eqref{eq:proofthm41_10m1} is larger than or equal to $\gamma$, and the term $(1-\gamma^{-1}) < 0$ for $0 < \gamma < 1$. Eq.~\eqref{eq:proofthm41_11} is valid because we do not apply clipping to $\rho_0^i$ and hence $\mathbb{E}_{a_0^i\sim\mu^i}[\rho_0^i]=\int \frac{\pi(a_0^i|o_0^i)}{\mu^i(a_0^i|o_0^i)} \mu^i(a_0^i|o_0^i) da_0^i=\int \pi(a_0^i|o_0^i)da_0^i=1$.  (Now, the reason why $\rho_0^i$ for the first action is not clipped is clear.)
Therefore, $\mathcal{R}^i$ is a $\gamma$-contraction. By the Banach fixed point theorem, $\mathcal{R}^i$ has a unique fixed point.

\textbf{Part II: Analysis on the general off-policy fixed point}

For the simplicity of the analysis, we assume that each agent well processes its historical data so that we can assume full-observability, i.e., $\boldsymbol{\pi}(\boldsymbol{a}|s) = \prod_{j=1}^N \pi^j (a^j|s)$, $\boldsymbol{\mu}(\boldsymbol{a}|s) = \prod_{j=1}^N \mu^j (a^j|s)$. In Agent $i$'s perspective, we regard $\tilde{y}^i_{t+1} = (s_{t+1}, \boldsymbol{a}_{t+1}^{-i})$ as an augmented state. Then we define a new transition function as follows:
\begin{equation}
    \mathcal{P}(s_{t+1}| s_{t}, \boldsymbol{a}_{t}) \prod_{j \neq i} \mu^j (a_{t+1}^j|s_{t+1}) =:  \mathcal{P}_i( \tilde{y}^i_{t+1} | \tilde{y}^i_{t}, a^i_t).
\end{equation}
In Agent $i$'s perspective, it interacts with the environment with the following augmented MDP: $<\mathcal{S}\times A^{N-1}, A, \mathcal{P}_i, r, \gamma>$, where $A$ is agent-wise action space and $A^{N} = \mathcal{A}$.

Consider the following term at timestep $t \geq 1$:
\begin{align}
    &\mathbb{E}_{\boldsymbol{\mu}} \Bigg[ \gamma^t \big(\prod_{j=1}^{t}c^{i}_j\big) \Big(r(s_t,\boldsymbol{a}_t) + \gamma  \overline{EQ}^i(s_{t+1},\boldsymbol{a}^{-i}_{t+1}) \nonumber\\
&\qquad\qquad\qquad\qquad\qquad\qquad- \overline{EQ}^i(s_t,\boldsymbol{a}^{-i}_t)\Big)\biggr| s_{0:t-1}, \boldsymbol{a}_{0:t-1} \Bigg] \nonumber \\
    &= \gamma^t \big(\prod_{j=1}^{t-1}c^{i}_j\big) \sum_{s_t} \mathcal{P}(s_t| s_{t-1}, \boldsymbol{a}_{t-1}) \sum_{\boldsymbol{a}_t} \prod_{j=1}^N \mu^j (a_t^j|s_t) c_t^i \times  \nonumber \\
    & \Big(r(s_t,\boldsymbol{a}_t) + \gamma  \sum_{s_{t+1}} \mathcal{P}(s_{t+1}| s_{t}, \boldsymbol{a}_{t}) \big( \sum_{ \boldsymbol{a}_{t+1}^{-i} } \prod_{j \neq i} \mu^j (a_{t+1}^j|s_{t+1})  \nonumber \\
    &\qquad\qquad\qquad\qquad\qquad \cdot\overline{EQ}^i(s_{t+1},\boldsymbol{a}^{-i}_{t+1}) \big) - \overline{EQ}^i(s_t,\boldsymbol{a}^{-i}_t)\Big) \nonumber \\
    &= \gamma^t \big(\prod_{j=1}^{t-1}c^{i}_j\big)  \sum_{\tilde{y}^i_{t}} \mathcal{P}_i( \tilde{y}^i_{t} |  \tilde{y}^i_{t-1}, a_{t-1}^i) \lambda C^i(\tilde{y}^i_{t}) \sum_{a_t^i} \pi_{*}^i(a_t^i | \tilde{y}^i_{t} ) \times  \nonumber \\
    & \left(r(\tilde{y}^i_{t}, a_t^i) + \gamma \sum_{\tilde{y}^i_{t+1}} \mathcal{P}_i( \tilde{y}^i_{t+1} | \tilde{y}^i_{t}, a_t^i)\overline{EQ}^i(\tilde{y}^i_{t+1})   - \overline{EQ}^i(\tilde{y}^i_{t})\right).
\end{align}

Here we insert $c_t^i = \lambda \min \big( 1, \frac{ \pi^i(a_t^i|s_t)  }{ \mu^i(a_t^i|s_t) } \min (1, \prod_{j \neq i} \frac{ \pi^j(a_t^j|s_t)  }{ \mu^j(a_t^j|s_t) }  )  \big)$
and define 
\begin{align}
    \pi_{*}^i(a_t^i | \tilde{y}^i_{t}) &= \pi_{*}^i(a_t^i | s_t, \boldsymbol{a}_{t}^{-i}) \nonumber \\
    &:= \frac{\mu^i (a_t^i|s_t) \min \big( 1, \frac{ \pi^i(a_t^i|s_t)  }{ \mu^i(a_t^i|s_t) } \min (1, \prod_{j \neq i} \frac{ \pi^j(a_t^j|s_t)  }{ \mu^j(a_t^j|s_t) }  )  \big) }{ \sum_{a'} \mu^i (a'|s_t) \min \big( 1, \frac{ \pi^i(a'|s_t)  }{ \mu^i(a'|s_t) } \min (1, \prod_{j \neq i} \frac{ \pi^j(a_t^j|s_t)  }{ \mu^j(a_t^j|s_t) }  )  \big)  }
\end{align}
for $t \geq 1$ and $\pi_{*}^i(a_0^i | \tilde{y}^i_{0}) = \mu^i(a_0^i | s_0)$. Plus, we denote the denominator as $C^i(\tilde{y}^i_{t}) = C^i(s_t, \boldsymbol{a}_{t}^{-i}) > 0$. Note that $\lambda(>0)$ does not appear in $\pi_{*}^i$.

By the definition of the Bellman equation for value function, the term
\begin{multline}
    \sum_{a_t^i} \pi_{*}^i(a_t^i | \tilde{y}^i_{t} ) \times \Bigg(r(\tilde{y}^i_{t}, a_t^i) + \gamma \sum_{\tilde{y}^i_{t+1}} \mathcal{P}_i( \tilde{y}^i_{t+1} | \tilde{y}^i_{t}, a_t^i)\overline{EQ}^i(\tilde{y}^i_{t+1}) \\
    - \overline{EQ}^i(\tilde{y}^i_{t})\Bigg)
\end{multline}
becomes zero if $\overline{EQ}^i(\tilde{y}^i_{t})$ is the value function of $\pi_{*}^i$ in the Agent $i$'s augmented MDP: $<\mathcal{S}\times A^{N-1}, A, \mathcal{P}_i, r, \gamma>$. Let the value function be $\phi^{\pi_{*}^i}$. Then 
\begin{equation}
    \mathcal{R}^i \overline{EQ}^{\pi_{*}^i} = \overline{EQ}^{\pi_{*}^i},
\end{equation}
so $\overline{EQ}^{\pi_{*}^i}$ is the unique fixed point of $\mathcal{R}^i$. 
\qed

{\em Remark:} Now we discuss the interpretation of $\phi^{\pi_{*}^i}$. Recall that 
\begin{align}
    \pi_{*}^i(a_t^i | \tilde{y}^i_{t}) &= \pi_{*}^i(a_t^i | s_t, \boldsymbol{a}_{t}^{-i}) \\
    &= \frac{ \min \big( \mu^i (a_t^i|s_t), \pi^i(a_t^i|s_t) \min (1, \boldsymbol{\rho}_t^{-i}(\boldsymbol{a}_{t}^{-i})  )  \big) }{ \sum_{a'} \min \big( \mu^i (a'|s_t), \pi^i(a'|s_t) \min (1,  \boldsymbol{\rho}_t^{-i}(\boldsymbol{a}_{t}^{-i})  )  \big) }
\end{align}
If $\boldsymbol{\rho}_t^{-i}(\boldsymbol{a}_{t}^{-i}) \geq 1$, $\pi_{*}^i(a_t^i | \tilde{y}^i_{t}) = \frac{ \min \big( \mu^i (a_t^i|s_t), \pi^i(a_t^i|s_t)  \big) }{ \sum_{a'} \min \big( \mu^i (a'|s_t), \pi^i(a'|s_t)   \big) }$, which is a conservative combination of $\pi^i$ and $\mu^i$. If $\boldsymbol{\rho}_t^{-i}(\boldsymbol{a}_{t}^{-i}) \to 0$, $\pi_{*}^i(a_t^i | \tilde{y}^i_{t}) $

\noindent$\to \pi^i(a_t^i|s_t)$. In short, $\pi_{*}^i(a_t^i | \tilde{y}^i_{t})$ is an approximation of Agent $i$'s current policy. 

Roughly speaking, the fixed point represents the average quality of the approximated target policy from Agent \(i\)'s perspective, considering all the other agents as part of the environment and following behavior policy \(\mu^j, j \neq i\). Then, \(\delta^{i,\text{GPAE}}_{t} = r_t + \gamma \overline{EQ}^i_{t+1} - \overline{EQ}^i_t\), using the fixed point, measures the contribution of \(a_t^i\) in this MDP.

If we control the behavior policy \(\boldsymbol{\mu}\) so that it does not deviate much from the current \(\boldsymbol{\pi}\), the off-policy fixed point does not deviate significantly from \(\mathbb{E}_{a^i \sim \pi^i}[Q^{\boldsymbol{\pi}}(s, a^i,  \boldsymbol{a}^{-i})]\). Then, \(\hat{A}_t^{i,\text{GPAE}}\) approximates GAE(\(\lambda\)) with an estimation of Agent \(i\)'s per-agent state value function. Note that \(\lambda\) does not affect the characteristic of the fixed point, but influences the per-agent advantage estimation. 

In the on-policy case, $\pi_{*}^i(a_t^i | \tilde{y}^i_{t}) = \pi^i(a_t^i|s_t)$, and the other agents take the current policy $\pi^j, j \neq i$. If we set  $Q = Q^{\boldsymbol{\pi}}$ in $\overline{EQ}^i$, we have
\begin{align}
    &\mathbb{E}_{\boldsymbol{\mu}  } \Bigg[ \gamma^t \big(\prod_{j=1}^{t}c^{i}_j\big) \left(r(s_t,\boldsymbol{a}_t) + \gamma  \overline{EQ}^i(s_{t+1},\boldsymbol{a}^{-i}_{t+1}) - \overline{EQ}^i(s_t,\boldsymbol{a}^{-i}_t)\right)\nonumber\\
    &\qquad\qquad\qquad\qquad\qquad\qquad\qquad\qquad\biggr| s_{0:t-1}, \boldsymbol{a}_{0:t-1} \Bigg] \nonumber \\
    &= (\gamma \lambda)^t \sum_{s_t} \mathcal{P}(s_t| s_{t-1}, \boldsymbol{a}_{t-1}) \sum_{\boldsymbol{a}_t} \prod_{j=1}^N \pi^j (a_t^j|s_t)  \times  \nonumber \\
    & \Bigg(r(s_t,\boldsymbol{a}_t) + \gamma  \sum_{s_{t+1}} \mathcal{P}(s_{t+1}| s_{t}, \boldsymbol{a}_{t}) \big( \sum_{ \boldsymbol{a}_{t+1}^{-i} } \prod_{j \neq i} \pi^j (a_{t+1}^j|s_{t+1}) \nonumber \\
    & \qquad\qquad\qquad\qquad\qquad\cdot \overline{EQ}^i(s_{t+1},\boldsymbol{a}^{-i}_{t+1}) \big) - \overline{EQ}^i(s_t,\boldsymbol{a}^{-i}_t)\Bigg) \nonumber \\
    &= (\gamma \lambda)^t \sum_{s_t} \mathcal{P}(s_t| s_{t-1}, \boldsymbol{a}_{t-1})  \sum_{\boldsymbol{a}_t} \prod_{j=1}^N \pi^j (a_t^j|s_t)  \times \nonumber \\
    &\qquad\qquad\qquad\qquad\qquad\left( Q^{\boldsymbol{\pi}}(s_t,\boldsymbol{a}_t) - \mathbb{E}_{a^i \sim \pi^i}[Q^{\boldsymbol{\pi}}(s_t, a^i,  \boldsymbol{a}_t^{-i})]  \right)  \nonumber \\
    &=  (\gamma \lambda)^t \sum_{s_t} \mathcal{P}(s_t| s_{t-1}, \boldsymbol{a}_{t-1})  (V^{\boldsymbol{\pi}}(s_t) - V^{\boldsymbol{\pi}}(s_t)) = 0.
\end{align}

Therefore, the fixed point of $\mathcal{R}^i$ is $\overline{EQ}^{\pi_{*}^i}(s_t, \boldsymbol{a}_t^{-i}) $

$= \mathbb{E}_{a^i \sim \pi^i}[Q^{\boldsymbol{\pi}}(s_t, a^i, \boldsymbol{a}_t^{-i})]$, and $\overline{EQ}^i$ is an estimation of the unique fixed point of $\mathcal{R}^i$, the Agent $i$'s per-agent state value function.

\vspace{2em}
\textbf{Part III: Showing Policy Invariance of $\hat{A}_t^{i,\text{GPAE}}$}

Now, consider the general off-policy case with the trace. In this case, we have the following:

\begin{align}
&\mathbb{E}_{\mathbf{\tau}\sim\boldsymbol{\mu}} \Bigg[\sum_{t\geq 0} \gamma^t \big(\prod_{j=1}^{t}c^{i}_j\big) \left(r(s_t,\boldsymbol{a}_t) + \gamma  \overline{EQ}^i(s_{t+1},\boldsymbol{a}^{-i}_{t+1}) - \overline{EQ}^i(s_t,\boldsymbol{a}^{-i}_t)\right)\nonumber\\
& \hspace{20em}\biggr| s_0=s, \boldsymbol{a}_0=\boldsymbol{a}\Bigg] \nonumber\\
&=\mathbb{E}_{\mathbf{\tau}\sim\boldsymbol{\mu}} \Bigg[\sum_{t\geq 0} \gamma^t \big(\prod_{j=1}^{t}c^{i}_j-\prod_{j=1}^{t}\boldsymbol{\rho}_j+\prod_{j=1}^{t}\boldsymbol{\rho}_j\big) \Big(r(s_t,\boldsymbol{a}_t) \nonumber\\
& \hspace{5em}+ \gamma  \overline{EQ}^i(s_{t+1},\boldsymbol{a}^{-i}_{t+1}) - \overline{EQ}^i(s_t,\boldsymbol{a}^{-i}_t)\Big)\biggr| s_0=s, \boldsymbol{a}_0=\boldsymbol{a}\Bigg]\label{eq:proofTh1_16} \\
&=\mathbb{E}_{\mathbf{\tau}\sim\boldsymbol{\mu}} \Bigg[\sum_{t\geq 0} \gamma^t \big(\prod_{j=1}^{t}\boldsymbol{\rho}_j\big) \Bigg(r(s_t,\boldsymbol{a}_t) + \gamma  \overline{EQ}^i(s_{t+1},\boldsymbol{a}^{-i}_{t+1}) \nonumber\\
& \hspace{10em}\left.- \overline{EQ}^i(s_t,\boldsymbol{a}^{-i}_t)\Bigg)\Bigg. \biggr| s_0=s, \boldsymbol{a}_0=\boldsymbol{a}\right] \label{eq:proofTh1_17} \\
&~~~~~+\mathbb{E}_{\mathbf{\tau}\sim\boldsymbol{\mu}} \Bigg[\sum_{t\geq 0} \gamma^t \big(\prod_{j=1}^{t}c^{i}_j-\prod_{j=1}^{t}\boldsymbol{\rho}_j\big) \Big(r(s_t,\boldsymbol{a}_t) + \gamma  \overline{EQ}^i(s_{t+1},\boldsymbol{a}^{-i}_{t+1}) \nonumber\\
& \hspace{10em}- \overline{EQ}^i(s_t,\boldsymbol{a}^{-i}_t)\Big)\biggr| s_0=s, \boldsymbol{a}_0=\boldsymbol{a}\Bigg] \label{eq:proofTh1_18} \\
&=\mathbb{E}_{\mathbf{\tau}\sim\boldsymbol{\pi}} \Bigg[\sum_{t\geq 0} \gamma^t  \Big(r(s_t,\boldsymbol{a}_t) + \gamma  \overline{EQ}^i(s_{t+1},\boldsymbol{a}^{-i}_{t+1}) - \overline{EQ}^i(s_t,\boldsymbol{a}^{-i}_t)\Big)\nonumber\\
& \hspace{15em}\biggr| s_0=s, \boldsymbol{a}_0=\boldsymbol{a}\Bigg] \label{eq:proofTh1_19} \\
&~~~~~+\mathbb{E}_{\mathbf{\tau}\sim\boldsymbol{\mu}} \Bigg[\sum_{t\geq 0} \gamma^t \big(\prod_{j=1}^{t}c^{i}_j-\prod_{j=1}^{t}\boldsymbol{\rho}_j\big) \Big(r(s_t,\boldsymbol{a}_t) + \gamma  \overline{EQ}^i(s_{t+1},\boldsymbol{a}^{-i}_{t+1}) \nonumber\\
& \hspace{10em} - \overline{EQ}^i(s_t,\boldsymbol{a}^{-i}_t)\Big)\biggr| s_0=s, \boldsymbol{a}_0=\boldsymbol{a}\Bigg] \label{eq:proofTh1_20} \\
&=  Q^{\boldsymbol{\pi}}(s_0, \boldsymbol{a}_0) -\mathbb{E}_{a^i\sim\pi^i}[Q(s_0,a_0^i,\boldsymbol{a}_0^{-i})] \label{eq:proofTh1_21}  \\
&~~~~~+  \mathbb{E}_{\mathbf{\tau}\sim\boldsymbol{\mu}} \left[\sum_{t\geq 0} \gamma^t \big(\prod_{j=1}^{t}c^{i}_j-\prod_{j=1}^{t}\boldsymbol{\rho}_j\big) r(s_t,\boldsymbol{a}_t) \biggr| s_0=s, \boldsymbol{a}_0=\boldsymbol{a}\right] ~~~~~~~~~~~(=: B_1)\label{eq:proofTh1_22}\\
&~~~~~+   \mathbb{E}_{\mathbf{\tau}\sim\boldsymbol{\mu}} \Bigg[\sum_{t\geq 1} \gamma^{t} \overline{EQ}^i(s_{t},\boldsymbol{a}_{t}^{-i})\Bigg((1-c_{t}^i)\prod_{j=1}^{t-1}c^{i}_j \nonumber \\
& \qquad\qquad\qquad\qquad\qquad - (1-\boldsymbol{\rho}_{t})\prod_{j=1}^{t-1}\boldsymbol{\rho}_j\Bigg)  \Bigg]  ~~~~~~~~~~~(=: B_2)\label{eq:proofTh1_23}\\
&~~~~~- ( 1-1)\overline{EQ}^i(s_0,\boldsymbol{a}_0^{-i}) \label{eq:proofTh1_24}\\
&=  Q^{\boldsymbol{\pi}}(s_0, \boldsymbol{a}_0) -\mathbb{E}_{a^i\sim\pi^i}[Q(s_0,a_0^i,\boldsymbol{a}_0^{-i})] + \underbrace{B_1 + B_2}_{=:B}.\label{eq:proofTh1_25}
\end{align}

Here, \eqref{eq:proofTh1_16} is obtained by subtracting and adding the same term $\prod_{j=1}^{t}\boldsymbol{\rho}_j$, where $\boldsymbol{\rho}_j$ is the unclipped original importance sampling ratio. \eqref{eq:proofTh1_17} and \eqref{eq:proofTh1_18} are obtained by separating \eqref{eq:proofTh1_16} into two terms. \eqref{eq:proofTh1_19} is valid because $\mathbb{E}_{\tau \sim \boldsymbol{\mu}}[\sum_t(\prod_{j=1}^{t}\boldsymbol{\rho}_j)(\cdot)]$ with unclipped trace is equivalent to $\mathbb{E}_{\tau \sim \boldsymbol{\pi}}[\sum_t (\cdot)]$ \cite{Precup2000}. Derivation of \eqref{eq:proofTh1_21} from \eqref{eq:proofTh1_19} is already done in the on-policy case. Two bias terms \eqref{eq:proofTh1_22} and \eqref{eq:proofTh1_23} are obtained by rearranging \eqref{eq:proofTh1_20}. \eqref{eq:proofTh1_24} is due to $\prod_{j=1}^0 c_j^i=\prod_{j=1}^0 \boldsymbol{\rho}_j=1$ by definition.  

Now consider the bias terms $B_1$ and $B_2$. Both terms reduces to zero if $c_j^i = \boldsymbol{\rho}$, thus $\hat{A}_t^{i,\text{GPAE}}$ guarantees policy invariance, which means $\hat{A}_t^{i,\text{GPAE}}$ is an unbiased advantage estimator. \qed

\vspace{2em}
\textbf{Part IV: Analysis on bias term of $\hat{A}_t^{i,\text{GPAE}}$}

The boundedness of $B_1$ can be shown as follows:
\begin{align}
B_1 &= \mathbb{E}_{\mathbf{\tau}\sim\boldsymbol{\mu}} \left[\sum_{t\geq 0} \gamma^t \big(\prod_{j=1}^{t}c^{i}_j-\prod_{j=1}^{t}\boldsymbol{\rho}_j\big) r(s_t,\boldsymbol{a}_t) \biggr| s_0=s, \boldsymbol{a}_0=\boldsymbol{a}\right] \nonumber\\
&= \mathbb{E}_{\mathbf{\tau}\sim\boldsymbol{\mu}} \left[\sum_{t\geq 0} \gamma^t \big(\prod_{j=1}^{t}c^{i}_j\big) r(s_t,\boldsymbol{a}_t) \biggr| s_0=s, \boldsymbol{a}_0=\boldsymbol{a}\right]\nonumber \\
&\qquad\qquad  - \mathbb{E}_{\mathbf{\tau}\sim\boldsymbol{\mu}} \left[\sum_{t\geq 0} \gamma^t \big(\prod_{j=1}^{t}\boldsymbol{\rho}_j\big) r(s_t,\boldsymbol{a}_t) \biggr| s_0=s, \boldsymbol{a}_0=\boldsymbol{a}\right]\nonumber \\
&= \mathbb{E}_{\mathbf{\tau}\sim\boldsymbol{\mu}} \left[\sum_{t\geq 0} \gamma^t \big(\prod_{j=1}^{t}c^{i}_j\big) r(s_t,\boldsymbol{a}_t) \biggr| s_0=s, \boldsymbol{a}_0=\boldsymbol{a}\right]\nonumber \\
&\qquad\qquad\qquad\qquad - \mathbb{E}_{\mathbf{\tau}\sim\boldsymbol{\pi}} \left[\sum_{t\geq 0} \gamma^t r(s_t,\boldsymbol{a}_t) \biggr| s_0=s, \boldsymbol{a}_0=\boldsymbol{a}\right]\label{eq:proofTh1_26}\\
&\le \mathbb{E}_{\mathbf{\tau}\sim\boldsymbol{\mu}} \left[\sum_{t\geq 0} \gamma^t \big(\prod_{j=1}^{t}c^{i}_j\big) (r(s_t,\boldsymbol{a}_t) + M_r) \biggr| s_0=s, \boldsymbol{a}_0=\boldsymbol{a}\right]\nonumber \\
&\qquad\qquad\qquad\qquad - \mathbb{E}_{\mathbf{\tau}\sim\boldsymbol{\pi}} \left[\sum_{t\geq 0} \gamma^t r(s_t,\boldsymbol{a}_t) \biggr| s_0=s, \boldsymbol{a}_0=\boldsymbol{a}\right]\label{eq:proofTh1_27}\\
&\le \mathbb{E}_{\mathbf{\tau}\sim\boldsymbol{\mu}} \left[\sum_{t\geq 0} \gamma^t  (r(s_t,\boldsymbol{a}_t) + M_r) \biggr| s_0=s, \boldsymbol{a}_0=\boldsymbol{a}\right] \nonumber \\
&\qquad\qquad\qquad\qquad- \mathbb{E}_{\mathbf{\tau}\sim\boldsymbol{\pi}} \left[\sum_{t\geq 0} \gamma^t r(s_t,\boldsymbol{a}_t) \biggr| s_0=s, \boldsymbol{a}_0=\boldsymbol{a}\right]\label{eq:proofTh1_28}\\
&= \mathbb{E}_{\mathbf{\tau}\sim\boldsymbol{\mu}} \left[\sum_{t\geq 0} \gamma^t  r(s_t,\boldsymbol{a}_t) \biggr| s_0=s, \boldsymbol{a}_0=\boldsymbol{a}\right]  
 + \frac{M_r}{1-\gamma} \nonumber \\
&\qquad\qquad\qquad\qquad- \mathbb{E}_{\mathbf{\tau}\sim\boldsymbol{\pi}} \left[\sum_{t\geq 0} \gamma^t r(s_t,\boldsymbol{a}_t) \biggr| s_0=s, \boldsymbol{a}_0=\boldsymbol{a}\right]\label{eq:proofTh1_29}\\
&= Q^{\boldsymbol{\mu}}(s,\boldsymbol{a}) - Q^{\boldsymbol{\pi}}(s,\boldsymbol{a}) + \frac{M_r}{1-\gamma} \label{eq:proofTh1_30}\\
&< \infty. \nonumber
\end{align}
Here, \eqref{eq:proofTh1_26} is valid because  $\mathbb{E}_{\tau \sim \boldsymbol{\mu}}[\sum_t(\prod_{j=1}^{t}\boldsymbol{\rho}_j)(\cdot)]=\mathbb{E}_{\tau \sim \boldsymbol{\pi}}[\sum_t (\cdot)]$ \cite{Precup2000}. \eqref{eq:proofTh1_27} is valid because the reward function is bounded, i.e., $-M_r \le r(s,a) \le M_r, \forall s,a$. 
\eqref{eq:proofTh1_28}
is valid because the clipped trace has  $0\le c_j^i \le 1$.  Note that when $r(s,a) \ge 0, \forall s,a$, we do not have the term $M_r/(1-\gamma)$ in \eqref{eq:proofTh1_29} and \eqref{eq:proofTh1_30}.

The boundedness of $B_2$ with finite $Q$-function can be shown as follows:
\begin{align}
B_2 &= \mathbb{E}_{\mathbf{\tau}\sim\boldsymbol{\mu}} \left[\sum_{t\geq 1} \gamma^{t} \overline{EQ}^i(s_{t},\boldsymbol{a}_{t}^{-i})\left((1-c_{t}^i)\prod_{j=1}^{t-1}c^{i}_j  - (1-\boldsymbol{\rho}_{t})\prod_{j=1}^{t-1}\boldsymbol{\rho}_j\right)  \right] \nonumber\\
&=\mathbb{E}_{\mathbf{\tau}\sim\boldsymbol{\mu}} \left[\sum_{t\geq 1} \gamma^{t} \overline{EQ}^i(s_{t},\boldsymbol{a}_{t}^{-i})\left((1-c_{t}^i)\prod_{j=1}^{t-1}c^{i}_j  \right)  \right] \nonumber\\
& ~~~~~-\mathbb{E}_{\mathbf{\tau}\sim\boldsymbol{\mu}} \left[\sum_{t\geq 1} \gamma^{t} \overline{EQ}^i(s_{t},\boldsymbol{a}_{t}^{-i})\left( (1-\boldsymbol{\rho}_{t})\prod_{j=1}^{t-1}\boldsymbol{\rho}_j\right)  \right] \\
&=\mathbb{E}_{\mathbf{\tau}\sim\boldsymbol{\mu}} \left[\sum_{t\geq 1} \gamma^{t} \overline{EQ}^i(s_{t},\boldsymbol{a}_{t}^{-i})\left((1-c_{t}^i)\prod_{j=1}^{t-1}c^{i}_j  \right)  \right] \label{eq:proofTh1_32}
\\
& ~~~~~-\mathbb{E}_{\mathbf{\tau}\sim\boldsymbol{\mu}} \left[\sum_{t\geq 1} \gamma^{t} \overline{EQ}^i(s_{t},\boldsymbol{a}_{t}^{-i})\left( \prod_{j=1}^{t-1}\boldsymbol{\rho}_j\right)  \right] \nonumber \\
&\qquad\qquad\qquad\qquad+\mathbb{E}_{\mathbf{\tau}\sim\boldsymbol{\mu}} \left[\sum_{t\geq 1} \gamma^{t} \overline{EQ}^i(s_{t},\boldsymbol{a}_{t}^{-i})\left(\boldsymbol{\rho}_{t}\prod_{j=1}^{t-1}\boldsymbol{\rho}_j\right)  \right] \label{eq:proofTh1_34}
\\
&=\mathbb{E}_{\mathbf{\tau}\sim\boldsymbol{\mu}} \left[\sum_{t\geq 1} \gamma^{t} \overline{EQ}^i(s_{t},\boldsymbol{a}_{t}^{-i})\left((1-c_{t}^i)\prod_{j=1}^{t-1}c^{i}_j  \right)  \right] \label{eq:proofTh1_35}
\\
& ~~~~~-\mathbb{E}_{\mathbf{\tau}\sim\boldsymbol{\mu}} \left[\sum_{t\geq 1} \gamma^{t} \overline{EQ}^i(s_{t},\boldsymbol{a}_{t}^{-i})\left( \prod_{j=1}^{t-1}\boldsymbol{\rho}_j\right)  \right] \nonumber \\
&\qquad\qquad\qquad\qquad +\mathbb{E}_{\mathbf{\tau}\sim\boldsymbol{\mu}} \left[\sum_{t\geq 1} \gamma^{t} \overline{EQ}^i(s_{t},\boldsymbol{a}_{t}^{-i})\left(\prod_{j=1}^{t}\boldsymbol{\rho}_j\right)  \right]\label{eq:proofTh1_37}
\\
&\le \mathbb{E}_{\mathbf{\tau}\sim\boldsymbol{\mu}} \left[\sum_{t\geq 1} \gamma^{t} \frac{M_r }{1-\gamma}  \right] \label{eq:proofTh1_38}\\
&+\mathbb{E}
_{\mathbf{\tau}\sim\boldsymbol{\mu}} \left[\sum_{t\geq 1} \gamma^{t} \frac{ M_r }{1-\gamma}\left( \prod_{j=1}^{t-1}\boldsymbol{\rho}_j\right)  \right] \label{eq:proofTh1_39}
\\
&+\mathbb{E}_{\mathbf{\tau}\sim\boldsymbol{\mu}} \left[\sum_{t\geq 1} \gamma^{t} \frac{ M_r }{1-\gamma}\left(\prod_{j=1}^{t}\boldsymbol{\rho}_j\right)  \right]\label{eq:proofTh1_40}
\\
&=\mathbb{E}_{\mathbf{\tau}\sim\boldsymbol{\mu}} \left[\sum_{t\geq 1} \gamma^{t} \frac{M_r }{1-\gamma}  \right] \label{eq:proofTh1_41}
\\
&+\gamma\mathbb{E}_{\mathbf{\tau}\sim\boldsymbol{\mu}} \left[\sum_{t\geq 1} \gamma^{t-1} \frac{ M_r }{1-\gamma}\left( \prod_{j=1}^{t-1}\boldsymbol{\rho}_j\right)  \right] \label{eq:proofTh1_42}
\\
&+\mathbb{E}_{\mathbf{\tau}\sim\boldsymbol{\mu}} \left[\sum_{t\geq 1} \gamma^{t} \frac{ M_r }{1-\gamma}\left(\prod_{j=1}^{t}\boldsymbol{\rho}_j\right)  \right]\label{eq:proofTh1_43}
\\
&=\mathbb{E}_{\mathbf{\tau}\sim\boldsymbol{\mu}} \left[\sum_{t\geq 1} \gamma^{t} \frac{M_r }{1-\gamma}  \right] \label{eq:proofTh1_44}
\\
&+\gamma\mathbb{E}_{\mathbf{\tau}\sim\boldsymbol{\mu}} \left[\sum_{t\geq 0} \gamma^{t} \frac{ M_r }{1-\gamma}\left( \prod_{j=1}^{t}\boldsymbol{\rho}_j\right)  \right] \label{eq:proofTh1_45}
\\
&+\mathbb{E}_{\mathbf{\tau}\sim\boldsymbol{\mu}} \left[\sum_{t\geq 0} \gamma^{t} \frac{ M_r }{1-\gamma}\left(\prod_{j=0}^{t}\boldsymbol{\rho}_j\right)  \right] - \frac{M_r}{1-\gamma}\label{eq:proofTh1_46}
\\
&= \frac{M_r\gamma}{(1-\gamma)^2} + \gamma \mathbb{E}_{\boldsymbol{\pi}} \left[ \sum_{t\ge 0} \gamma^t \frac{M_r}{1-\gamma}\right] + \mathbb{E}_{\boldsymbol{\pi}} \left[ \sum_{t\ge 0} \gamma^t \frac{M_r}{1-\gamma} \right] - \frac{M_r}{1-\gamma} \label{eq:proofTh1_47} 
\\
&= \frac{M_r\gamma}{(1-\gamma)^2} +\frac{M_r\gamma}{(1-\gamma)^2}+ \frac{M_r}{(1-\gamma)^2} - \frac{M_r}{1-\gamma} < \infty. \nonumber
\end{align}
Here, \eqref{eq:proofTh1_37} is obtained by combining $\boldsymbol{\rho}_t$ with $\prod_{j=1}^{t-1}\boldsymbol{\rho}_j$ in \eqref{eq:proofTh1_34}. \eqref{eq:proofTh1_38} is obtained from \eqref{eq:proofTh1_35} because $0\le c_j^i \le 1$ and $|\overline{EQ}^i(s,\boldsymbol{a}^{-i})| \le \frac{M_r}{1-\gamma}$. Note that since the reward function is bounded, i.e., $|r(s,a)| \le M_r, \forall s,a$, any value, i.e., a discounted sum of rewards is bounded as $-\frac{M_r}{1-\gamma}=-\sum_{t\ge 0} \gamma^t M_r \le \sum_{t\ge 0} \gamma^t r_t \le \sum_{t\ge 0} \gamma^t M_r= \frac{M_r}{1-\gamma}$. Hence, $\overline{EQ}^i(s,\boldsymbol{a}^{-i}) = \mathbb{E}_{a^i\sim\pi^i}[Q(s,a^i,\boldsymbol{a}^{-i})]$ is bounded by $\frac{M_r}{1-\gamma}$. \eqref{eq:proofTh1_42} is obtained from \eqref{eq:proofTh1_39}. \eqref{eq:proofTh1_45} and \eqref{eq:proofTh1_46} are the rearrangement of the sum indices of \eqref{eq:proofTh1_42} and \eqref{eq:proofTh1_43}, respectively. \eqref{eq:proofTh1_47} is obtained by using  $\mathbb{E}_{\tau \sim \boldsymbol{\mu}}[\sum_t(\prod_{j=1}^{t}\boldsymbol{\rho}_j)(\cdot)]=\mathbb{E}_{\tau \sim \boldsymbol{\pi}}[\sum_t (\cdot)]$ \cite{Precup2000}.






\newpage

\section{Additional Explanation for the Motivation Experiment}
\label{appendix:motivation}

This section provides a more detailed quantitative interpretation of the advantage gap metric introduced in Section~\ref{subsec:motivation_exp}. The analysis clarifies how the SMAX environment's reward structure and dynamics permit an approximate theoretical reference for \(\Delta A\) under the controlled perturbation setting used in the motivation experiment.

\vspace{0.35em}
\noindent \textbf{Environment and perturbation setup.}  
In the SMAX-\texttt{3m} task, the team reward consists of two parts: (1) a cumulative \emph{damage reward} proportional to the total reduction in enemy HP, and (2) a binary \emph{win reward} of \(1.0\) obtained when all opponents are defeated. To create a controlled perturbation, one of the three agents is forced to execute a \texttt{stop} action with probability \(0.05\) at random timesteps, interrupting its attack and consequently lowering both its damage contribution and the team’s win probability.

\vspace{0.35em}
\noindent \textbf{(1) Loss in cumulative damage reward due to abnormal action.}  
Each marine has \(45\) HP and inflicts \(9\) HP of damage per attack, with a firing rate of approximately \(0.82\) attacks per timestep. Since the total enemy HP pool is \(135 = 45\times3\), each marine contributes an expected reward of
\[
\frac{9}{135}\times0.82 \approx 0.0547
\]
per timestep under normal attacking behavior. It takes approximately \(135/9=15\) attacks to eliminate all enemies, and with three marines attacking in parallel at a rate of \(0.82\), the expected number of timesteps required is \(15/3/0.82\approx6.1\). Given that an average episode lasts about \(30\) timesteps, each agent spends roughly \(6.1/30\approx0.203\) of its episode performing active attacks. Therefore, if a \texttt{stop} action replaces an attack during this period, the expected damage reward loss is approximately \(0.203\times0.0547\approx0.011\).

For a normal action \(a\), the empirical advantage can be expressed as
\[
A(s,a)= Q(s,a) - V(s) = r_1(s,a)+\gamma r_2+\gamma^2 r_3+\cdots+\gamma^T r_{T+1}-V(s).
\]
When an abnormal action \(\tilde{a}\) is executed, the immediate reward \(r_1(s,\tilde{a})\) is smaller by about \(0.011\), while the subsequent rewards \(r_2, r_3, \ldots\) remain statistically similar. Hence, the value gap from the damage component is roughly \(0.011\).

\vspace{0.35em}
\noindent \textbf{(2) Loss in win-reward value due to abnormal action.}  
Empirically, the baseline win rate in \texttt{3m} is approximately \(60\%\). Under the perturbation with a \(5\%\) anomaly probability, the win rate drops by about \(10\%\), meaning that one in ten episodes that would have succeeded now fails due to the anomaly. Since the win reward of \(1.0\) is only given upon victory, we can model the expected value difference as
\[
\gamma^{15}\times 0.6\times(1-0.9) = 0.052,
\]
where \(\gamma = 0.99\) and we assume the abnormal action occurs around \(t=15\), the middle of a typical 30-step episode.

\vspace{0.35em}
\noindent \textbf{(3) Theoretical reference and empirical comparison.}  
Summing both components, the approximate theoretical advantage gap becomes
\[
\Delta A_{\text{theory}} \approx 0.063 = 0.052 + 0.011.
\]
In our experiments, GPAE-off produced an empirical \(\Delta A \approx 0.013\), which, while smaller than the theoretical upper bound, is substantially closer to the ideal value than those of other estimators (GAE: \(0.000\), DAE: \(0.0014\), COMA: \(0.0016\)). Moreover, GPAE yields near-zero \(\Delta A\) in trajectories without anomalies, confirming that it does not introduce artificial variance.

\vspace{0.35em}
\noindent \textbf{Discussion.}  
This analysis provides a quantitative grounding for the diagnostic experiment in Section~\ref{sec:motivation}. The controlled perturbation allows a tractable approximation of an ``ideal'' per-agent value gap derived from the environment’s reward mechanics. The close alignment of GPAE with this theoretical reference supports that the estimator accurately captures credit differences caused by anomalous actions rather than introducing spurious variance or bias.

\section{More Explanation on DT-ISR}
\label{appendix:dt}

This section provides a more detailed explanation of the proposed double-truncated importance sampling ratio (DT-ISR) weight, which is introduced to address the variance explosion problem in multi-agent off-policy learning. We first discuss the key elements that define the DT-ISR, including the individual and joint importance sampling ratios, and then explore the effects of the truncation mechanisms on agent learning.

The DT-ISR weight at time step $t$ for agent $i$ is formulated as:

\[
c_t^{i,\text{DT}} = \lambda \min\left(1, \rho_t^i \min(\eta, \boldsymbol{\rho}_t^{-i})\right)
\]

Here, \(\rho_t^i\) is the individual ISR for agent \(i\), \(\boldsymbol{\rho}_t^{-i}\) is the joint ISR excluding agent \(i\), and \(\eta\) is a constant that limits the impact of other agents’ policies.

\subsection*{Function of $c_t^{i,\text{DT}}$ as a function of $\rho^i$ and $\boldsymbol{\rho}^{-i}$}

The weight $c_t^{i,\text{DT}}$ is designed to adapt based on two components:

\begin{enumerate}
    \item \textbf{Individual ISR \(\rho_t^i\)}: This component captures the relative behavior of agent \(i\) with respect to its own policy and behavior distribution. It is defined as the ratio of the policy \(\pi_t^i\) of agent \(i\) over the behavior policy \(\mu_t^i\), i.e.,
    \[
    \rho_t^i = \frac{\pi_t^i}{\mu_t^i}
    \]
    The term \(\rho_t^i\) reflects how well agent \(i\) is following its own policy relative to the distribution of the observed behavior. When \(\rho_t^i > 1\), the agent is acting more greedily (relative to its behavior policy), while \(\rho_t^i < 1\) indicates suboptimal or exploratory behavior.
    
    \item \textbf{Joint ISR excluding agent \(i\), \(\boldsymbol{\rho}_t^{-i}\)}: This term considers the collective behavior of all agents except \(i\), capturing how other agents' actions influence the environment. The joint ISR \(\boldsymbol{\rho}_t\) is computed as the product of all agents’ individual ISRs, i.e.,
    \[
    \boldsymbol{\rho}_t = \prod_{j \in \mathcal{N}} \frac{\pi_t^j}{\mu_t^j}
    \]
    The exclusion of agent \(i\) in \(\boldsymbol{\rho}_t^{-i}\) ensures that the effect of other agents’ policies is considered without overemphasizing their contribution in the truncation mechanism.
\end{enumerate}

The \textbf{DT-ISR} weight balances these two components by truncating the joint ISR at a factor \(\eta\), which serves to prevent the joint dynamics from interfering too strongly with agent \(i\)’s contribution. The function \(\min(\eta, \boldsymbol{\rho}_t^{-i})\) ensures that the impact of other agents is moderated while maintaining the individual ISR’s importance. This enables the agent to focus more effectively on its own actions without being overwhelmed by the collective behavior of other agents.

\subsection*{Impact on Learning}

\begin{itemize}
    \item \textbf{When \(\rho_t^i\) is large (close to 1 or greater)}: As \(\rho_t^i\) increases, the term \(\lambda \min(1, \rho_t^i)\) becomes more significant, giving greater weight to agent \(i\)’s own actions. This increases the agent's control over its learning process, making the learning more sensitive to its policy rather than the behavior of other agents. However, as the agent becomes too focused on its own policy, it might fail to generalize well to changes in other agents' policies, leading to suboptimal learning in dynamic multi-agent settings.
    
    \item \textbf{When \(\boldsymbol{\rho}_t^{-i}\) is large}: If the joint ISR (excluding agent \(i\)) is large, it indicates that the actions of other agents are heavily influencing the environment. The \(\min(\eta, \boldsymbol{\rho}_t^{-i})\) term helps mitigate the effect of large ISR values from other agents by truncating them to \(\eta\), thus ensuring that the agent’s updates remain focused on its own policy rather than on the joint dynamics of all agents. This helps stabilize learning in environments where other agents' behavior may be changing rapidly, thus improving robustness and reducing variance.
\end{itemize}

The introduction of the constant \(\eta\) plays a crucial role in preventing large fluctuations in the weight due to other agents' policies. By limiting $\boldsymbol{\rho}_t^{-i})$ to a value less than or equal to \(\eta\), the DT-ISR prevents instability caused by sudden shifts in other agents' behaviors while still allowing for the preservation of agent \(i\)'s policy dynamics.

\section{Implementation Details}
\label{appendix:implementation}

This section involves three key components of our implementation: value estimation, policy learning, and off-policy sample reuse. Each component is designed to leverage the theoretical advantages of the $\mathcal{R}^i$ operator and GPAE. A brief flow structure can be found in the Fig. \ref{fig:structure}.  The code for GPAE is
available at \texttt{\href{https://github.com/kim-seongmin/GPAE}{https://github.com/kim-seongmin/GPAE}}.

\subsection{Value Estimation}

The value network takes three inputs: the state $s$, joint action of other agent $\boldsymbol{a}^{-i}$, and the action probabilities from the policy $\pi_\theta$ and observation $o^i$. Importantly, instead of directly learning the state-action value function $Q(s, \boldsymbol{a})$, the network indirectly estimates it by substituting the action $a^i$ with $\pi_\theta(\cdot|o^i)$. This design choice is underpinned by the contraction property guaranteed by the $\mathcal{R}^i$ operator.

\begin{figure}[h]
    \centering
    \includegraphics[width=\linewidth]{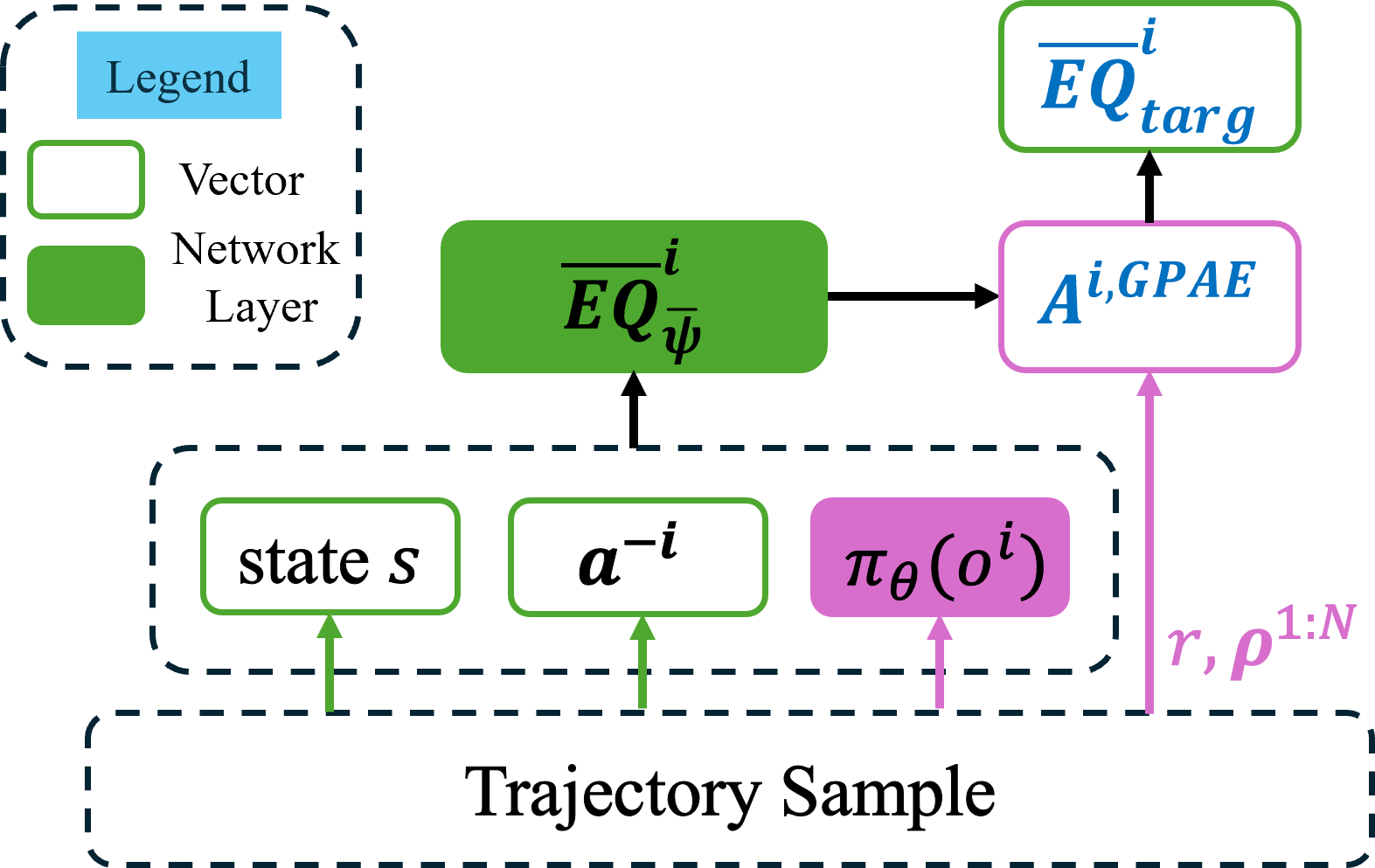}
    \caption{Brief overview of GPAE framework}
    \Description{Brief overview of GPAE framework}
    \label{fig:structure}
\end{figure}

We construct separate hidden layers for each pair of state and marginalized action dimensions. These layers are concatenated in the hidden dimension to produce the final output. The first layer $\psi_1$ represents a linear network layer with latent for state as output, $\psi_2$ is a layer for marginalized action, and $\psi_3$ exports $\overline{EQ}^i_{\overline{\psi}}$ as the final output through GRU, with concatenation of the previous two features.
To ensure a fair comparison, the hidden layer size for each state-action pair is halved, so that after concatenation, the total hidden layer size matches the original configuration.

The value function is optimized using the following loss function:
\begin{equation}
\label{eq:Qloss}
    L(\psi) = \sum_{\tau \in \mathcal{B}} (\mathbb{E}_{\pi^i}[\overline{EQ}^i_{\psi}] - \overline{EQ}_{targ}^i)^2,
\end{equation}
where $\mathcal{B}$ represents the replay buffer, and the target value $\overline{EQ}_{targ}^i$ is explicitly given   by 
\begin{equation}
\label{eq:value_estimation}
    \overline{EQ}_{targ}^i = \mathbb{E}_{\pi^i}[\overline{EQ}^i_{\overline{\psi}}] + \overline{\rho_t^i}\hat{A}_t^{i,GPAE},
\end{equation}

where $\overline{\rho_t^i} = \min(1, \rho_t^i)$ is a bounded individual ISR for stable learning.

\subsection{Off-policy Sample Reuse and Policy Learning}

Although our framework can be extended to general MAPG methods, we configured the policy learning structure identical to MAPPO for consistency and fair comparison. The policy optimization process employs a clipped surrogate objective from MAPPO \cite{mappo}, while using $\hat{A}_t^{i,\text{GPAE}}$ for the advantage function.

To maximize sample efficiency, the proposed framework incorporates off-policy sample reuse strategies. This approach aligns closely with methods employed in several single-agent off-policy PPO frameworks\cite{disc,geppo,offpolicyppo}. For off-policy sample reuse, we consider the behavioral old policy $\pi_\text{old}$ and current target policy $\pi_\text{curr}$ separately. 

We use the following off-policy policy learning loss:
\begin{multline}
\label{eq:policy_loss}
L_\theta = \mathbb{E}_{\tau \sim \mathcal{B}} \biggl[ \min \biggl( \rho_t^i \hat{A}_t^{i,\text{GPAE}}, 
\text{clip}(\rho_t^i, \rho^i_{t,\text{old}}(1 - \epsilon), \\
\rho^i_{t,\text{old}}(1 + \epsilon)) \hat{A}_t^{i,\text{GPAE}} \biggr) \biggr],
\end{multline}
where $\rho_t^i = \pi^i_\theta/\pi^i_\text{old}$ represents the ISR between the updating policy and the old policy, and $\rho^i_{t,\text{old}}=\pi^i_\text{curr}/\pi^i_\text{old}$ denotes the ratio between the current policy and the old policy.




\newpage

\section{Environment Details}
\label{appendix:env_details}
\subsection{SMAX}
SMAX is a JAX-based, hardware-accelerated reimplementation of the StarCraft Multi-Agent Challenge (SMAC) that removes the dependency on the StarCraft II engine. Unlike SMAC, which is bound to CPU-based simulation with limited parallelism, SMAX supports efficient GPU acceleration, significantly improving simulation speed and scalability.

The environment consists of multiple cooperative agents engaging in combat against an adversarial team, where the goal is to eliminate all opponents. Agents operate under partial observability, receiving local observations of nearby units within a predefined sight range. These observations include ally/enemy health, positions, actions, and unit types.

SMAX supports a discrete action space with movement and attack actions, similar to SMAC, but features more configurable dynamics and a flexible scenario design. Notably, the built-in heuristic AI for adversaries is decentralized and reactive, agents pursue visible enemies and maintain consistent targeting behavior across steps. This contrasts with SMAC's globally conditioned or static scripted agents, making SMAX more suitable for self-play and adversarial robustness studies.

The reward structure in SMAX is also rebalanced: 50\% of the total return is given for damaging enemies, and 50\% is awarded for winning the match, making performance metrics more interpretable and stable across different scenarios.

Table~\ref{tab:smax_tasks} summarizes the available scenarios, including direct ports of SMAC(v1) maps and several stochastic setups inspired by SMACv2.

\begin{table}[h]
\centering
\caption{Evaluated tasks in SMAX. These scenarios involve discrete actions and cooperative combat against adversarial teams.}
\label{tab:smax_tasks}
\begin{tabular}{lll}
\toprule
\textbf{Task Name} & \textbf{Ally Units} & \textbf{Enemy Units} \\
\midrule
2s3z & 2 Stalkers, 3 Zealots & 2 Stalkers, 3 Zealots \\
3s5z & 3 Stalkers, 5 Zealots & 3 Stalkers, 5 Zealots \\
5m\_vs\_6m & 5 Marines & 6 Marines \\
10m\_vs\_11m & 10 Marines & 11 Marines \\
3s5z\_vs\_3s6z & 3 Stalkers, 5 Zealots & 3 Stalkers, 6 Zealots \\
6h\_vs\_8z & 6 Hydralisks & 8 Zealots \\
smacv2\_5\_units & 5 Random Units & 5 Random Units \\
smacv2\_10\_units & 10 Random Units & 10 Random Units \\
\bottomrule
\end{tabular}
\end{table}

\subsection{MABrax}
MABrax is a continuous control multi-agent environment built on top of Brax \cite{brax}, extending the concept of Multi-Agent MuJoCo \cite{facmac} to fully leverage JAX's acceleration capabilities. Each agent controls an individual joint of a shared robot (e.g., Ant, HalfCheetah, Humanoid), and they must coordinate their actions to complete locomotion tasks efficiently.

The continuous action space requires fine-grained control and synchronization between agents, highlighting the challenge of accurate credit assignment. The centralized value function is crucial here to guide each agent's learning toward coherent joint behaviors. Partial observability is implemented by restricting each agent's perception to its own joint state and neighboring components, encouraging policies that rely on decentralized coordination.

Unlike single-agent Brax environments, MABrax introduces multi-agent learning dynamics where poor credit assignment can easily destabilize training. Thus, it serves as an ideal testbed for evaluating credit assignment techniques under continuous, partially observable settings.

\section{Additional Experimental Results}
\label{appendix:experiments}

\subsection{Task-wise Learning Curve}
\label{appendix:taskwise_performance}

Fig.~\ref{fig:onpolicy_smax} presents the learning curves for each task in the main results.

\begin{figure*}[b]
\centering    
    \includegraphics[width=0.8\textwidth]{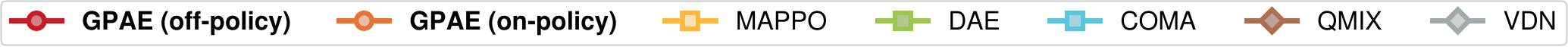}
    
    \centering\includegraphics[width=0.25\textwidth]{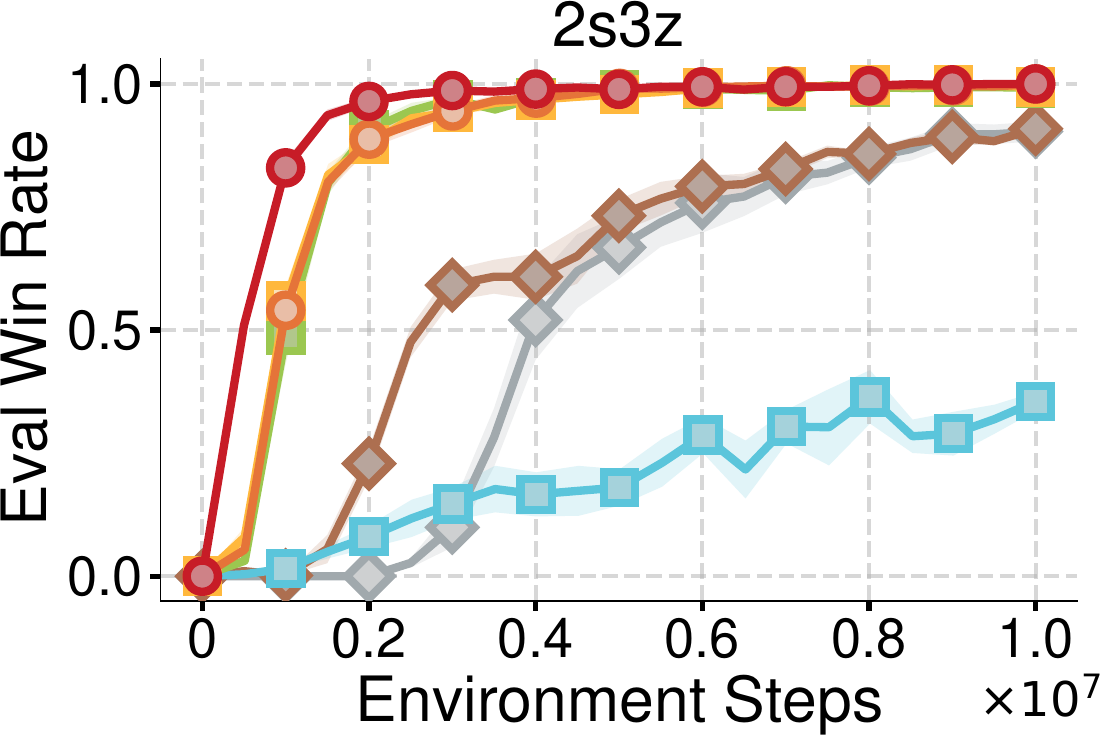}%
    \centering\includegraphics[width=0.25\textwidth]{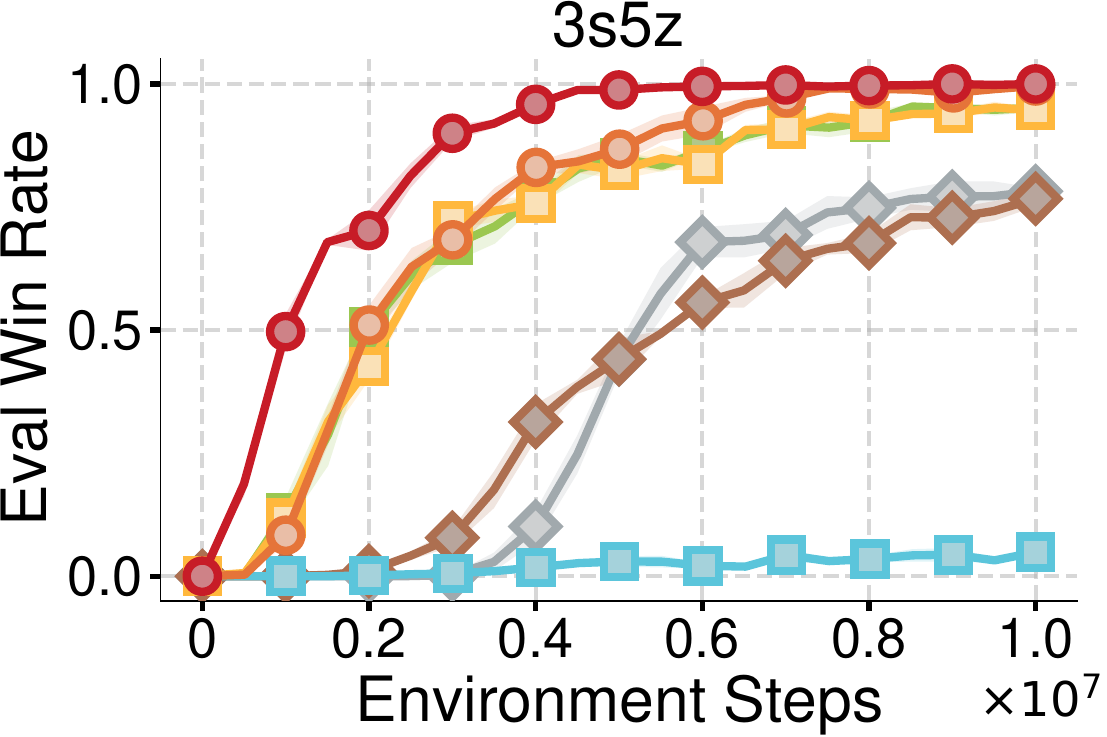}%
    \centering\includegraphics[width=0.25\textwidth]{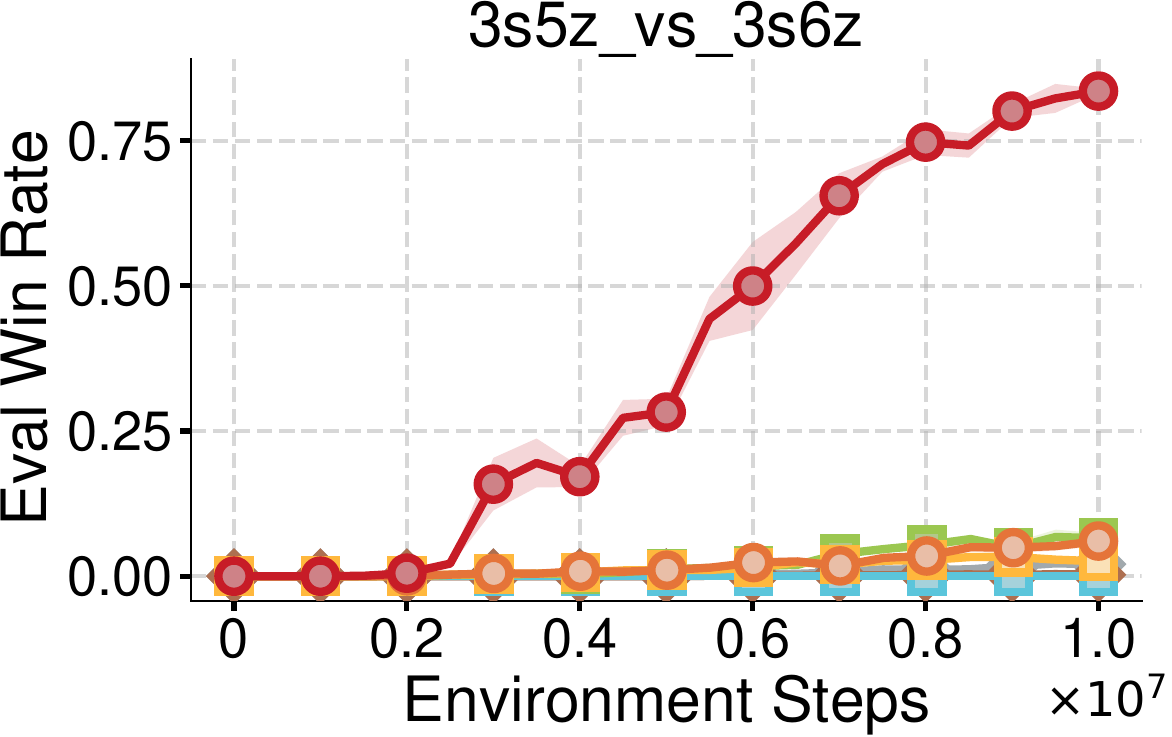}%
    \centering\includegraphics[width=0.25\textwidth]{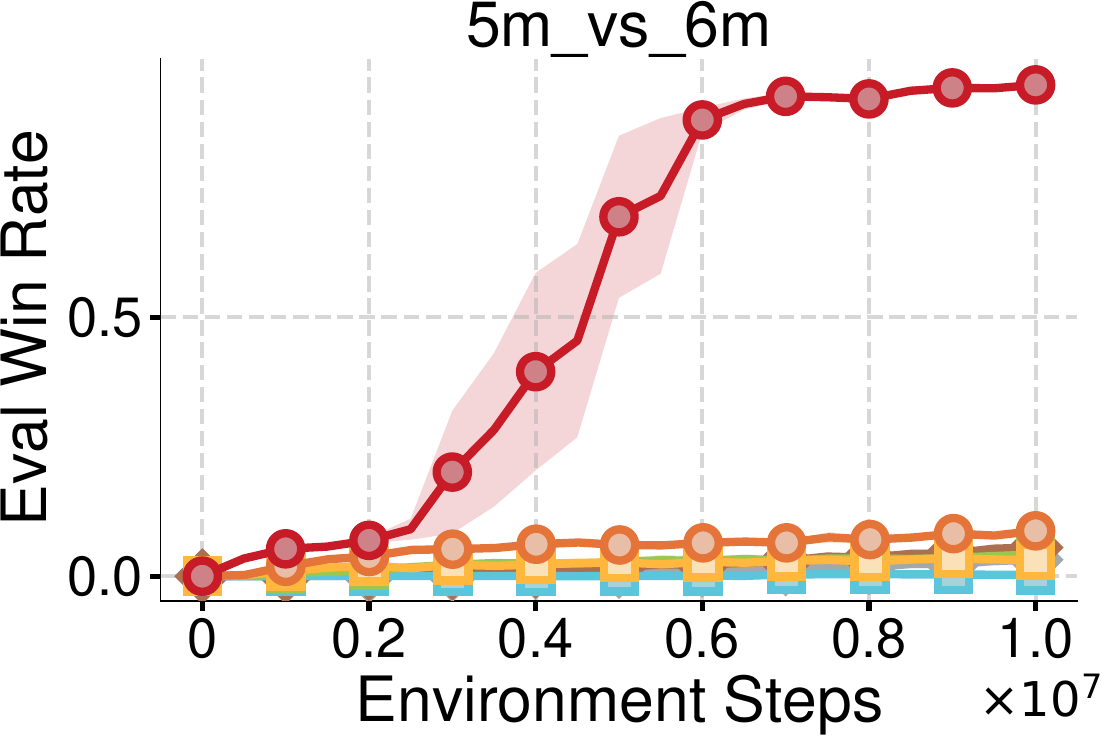}
    
    \centering\includegraphics[width=0.25\textwidth]{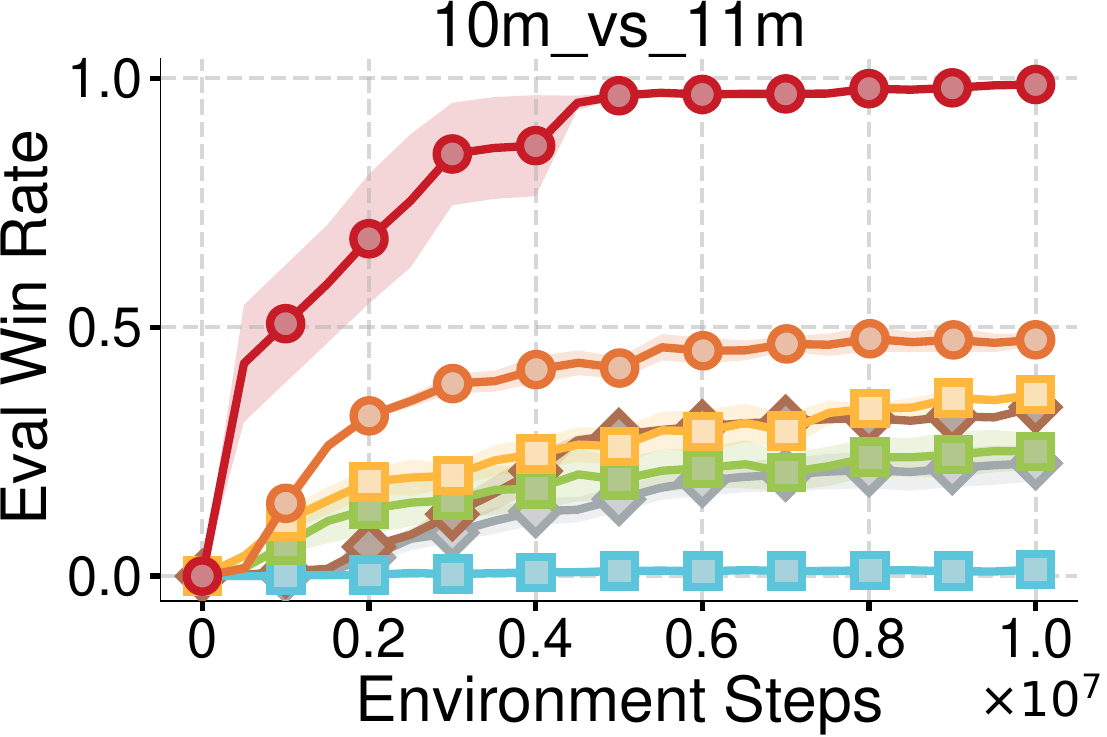}%
    \centering\includegraphics[width=0.25\textwidth]{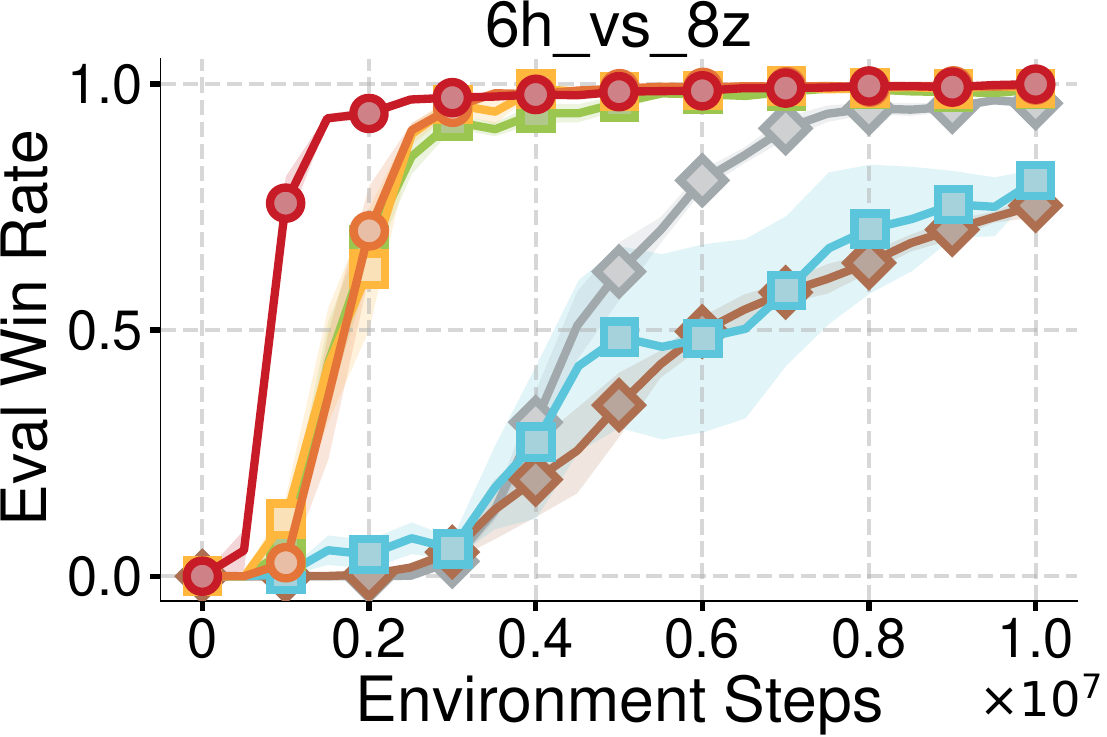}%
    \centering\includegraphics[width=0.25\textwidth]{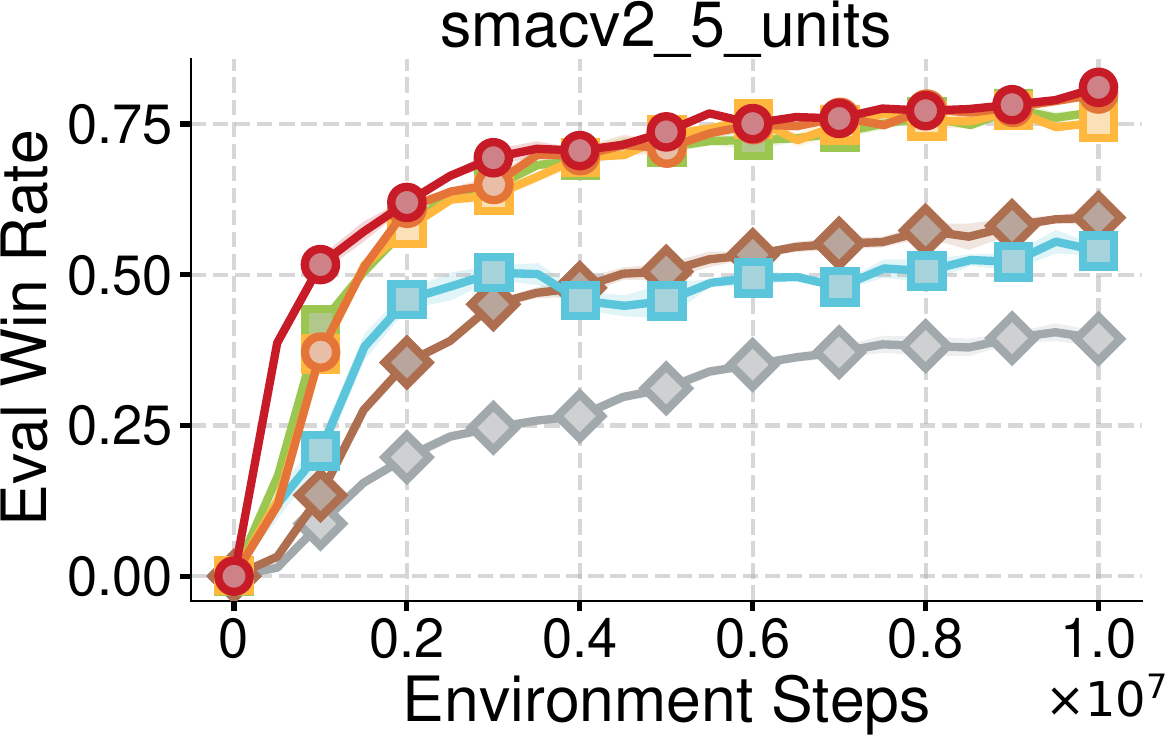}%
    \centering\includegraphics[width=0.25\textwidth]{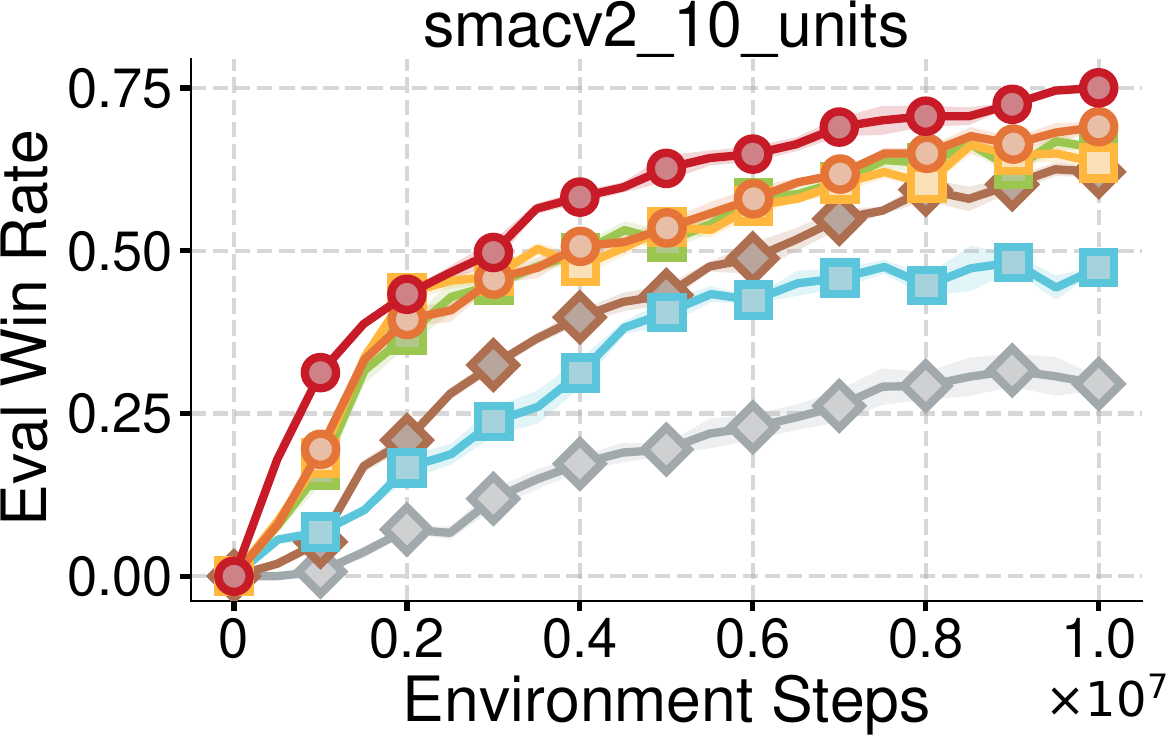}%
    
    \centering\includegraphics[width=0.25\textwidth]{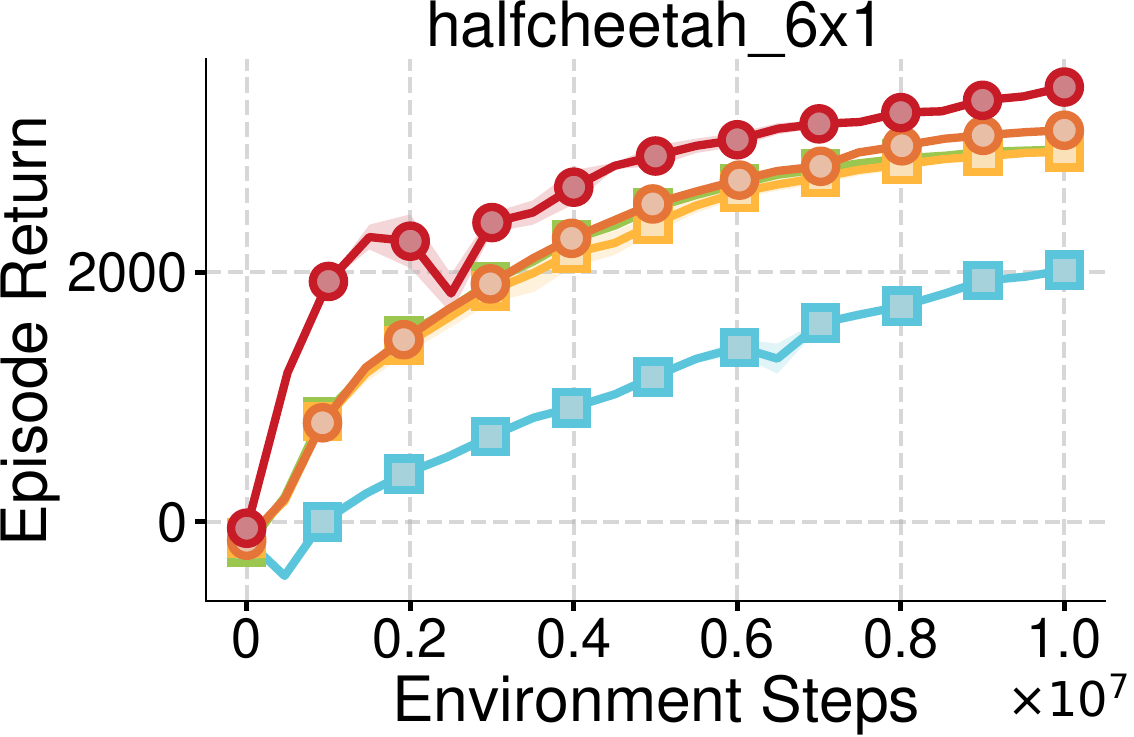}%
    \centering\includegraphics[width=0.25\textwidth]{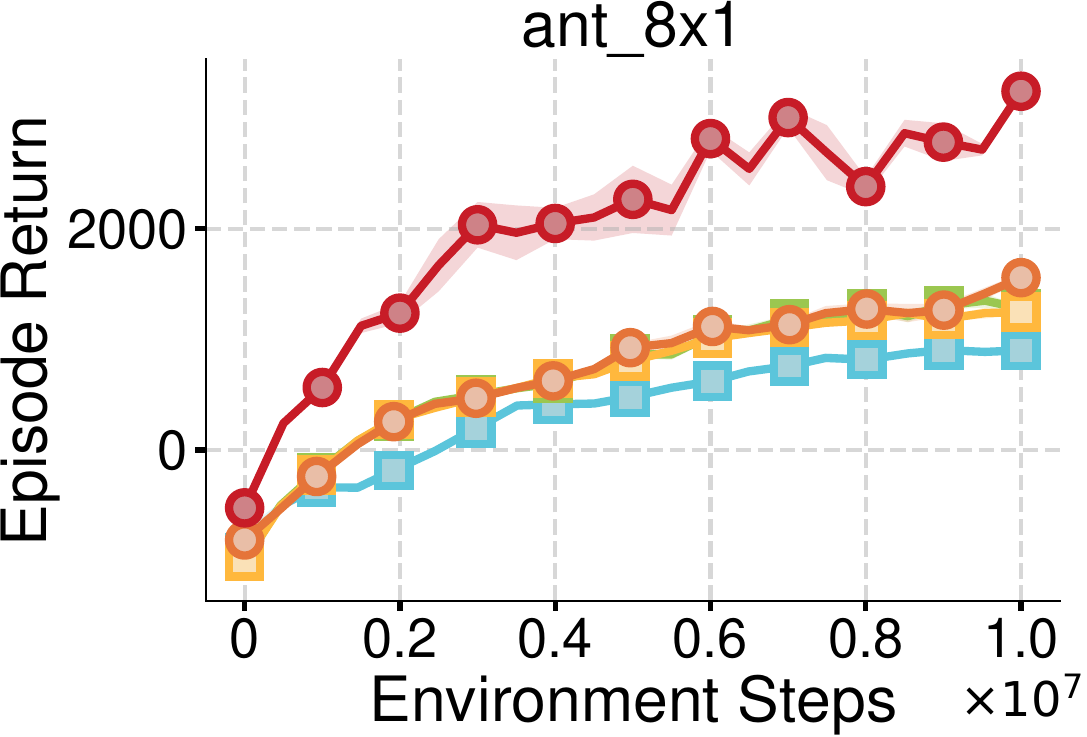}%
    \centering\includegraphics[width=0.25\textwidth]{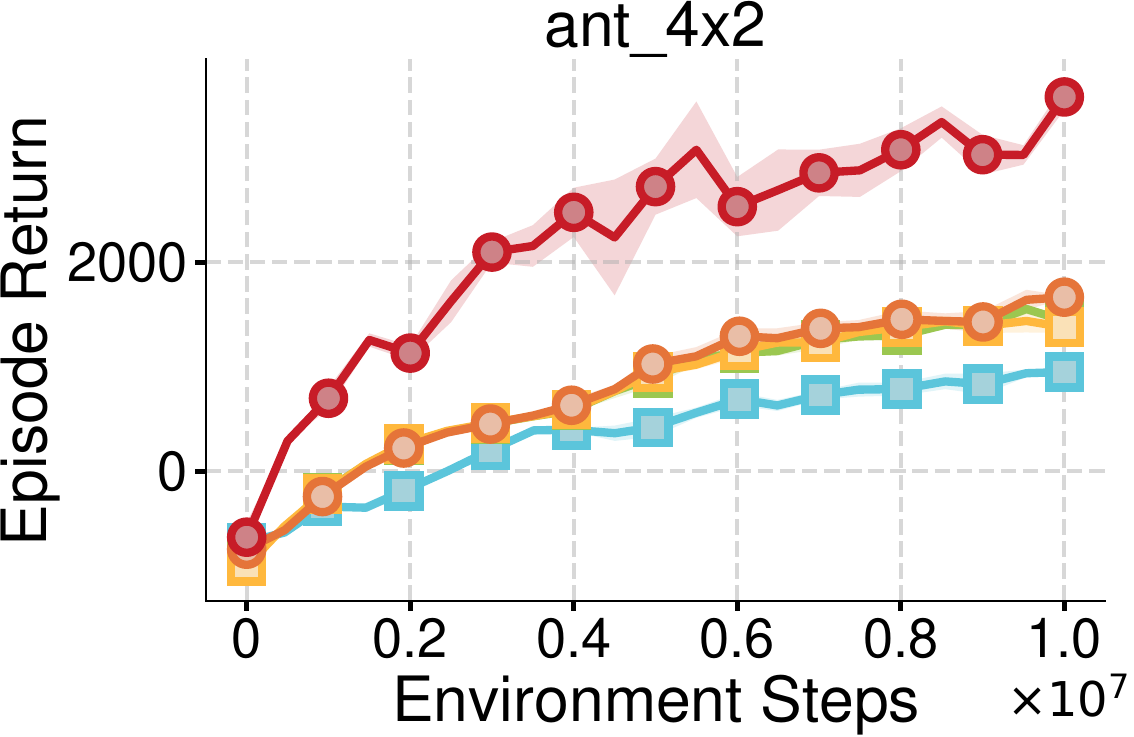}
    
    \centering\includegraphics[width=0.25\textwidth]{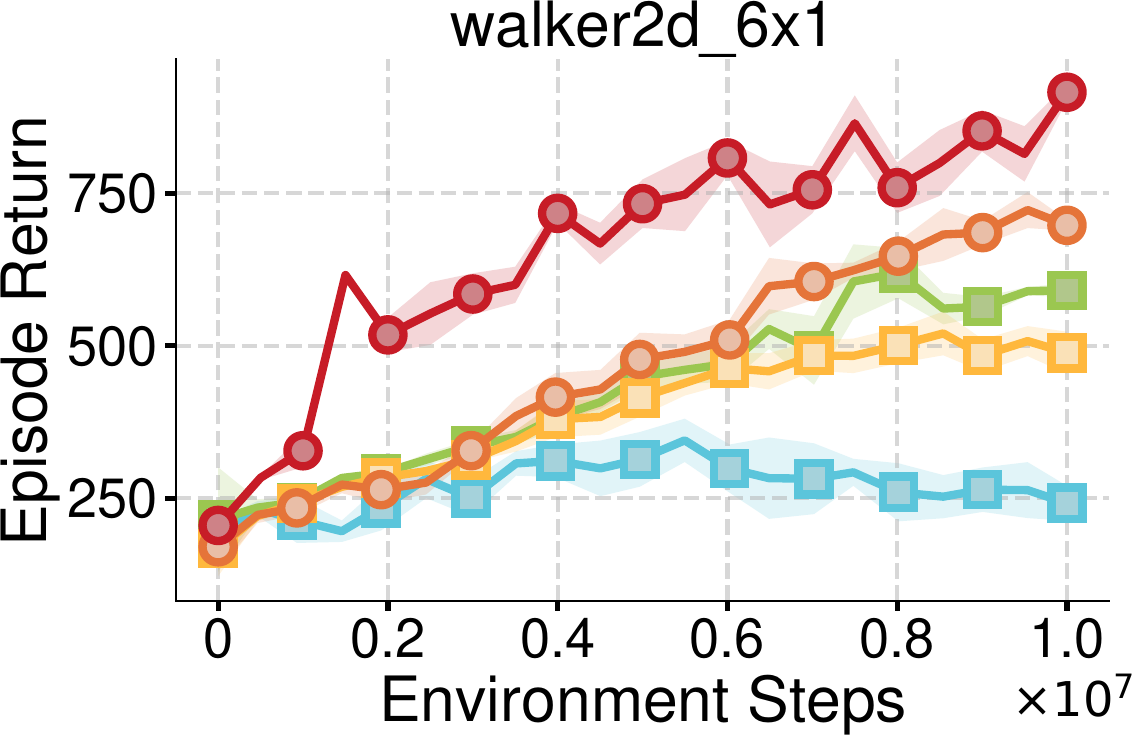}%
    \centering\includegraphics[width=0.25\textwidth]{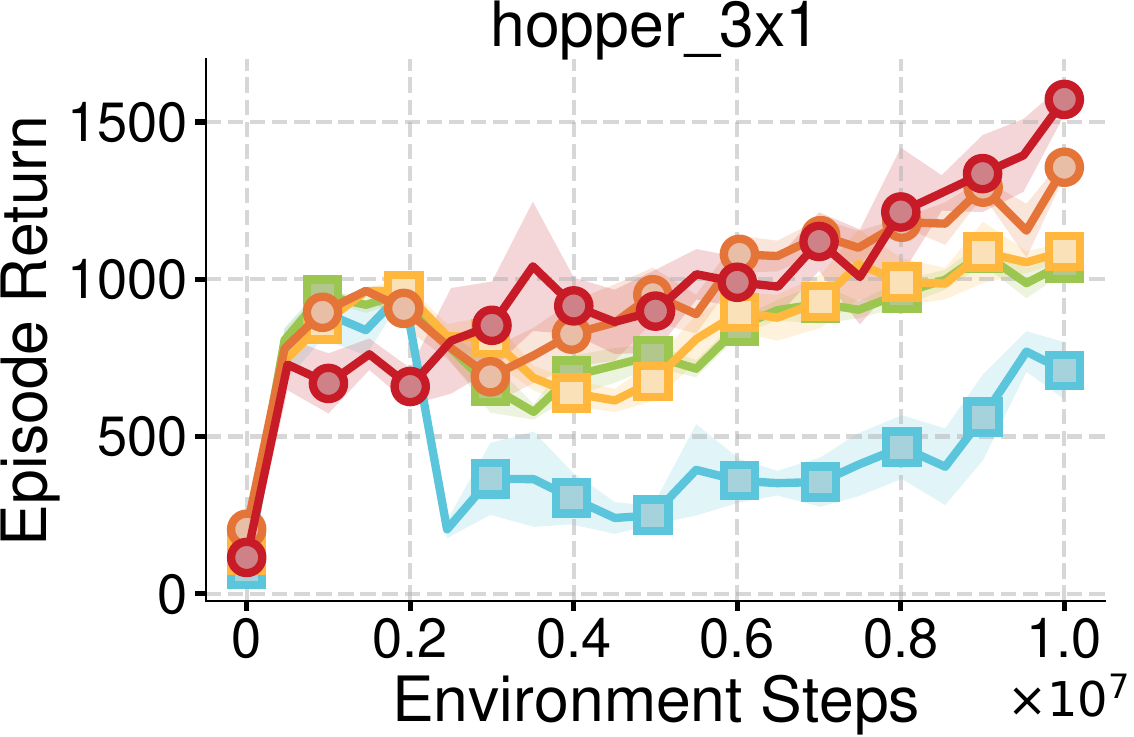}%
    \centering\includegraphics[width=0.25\textwidth]{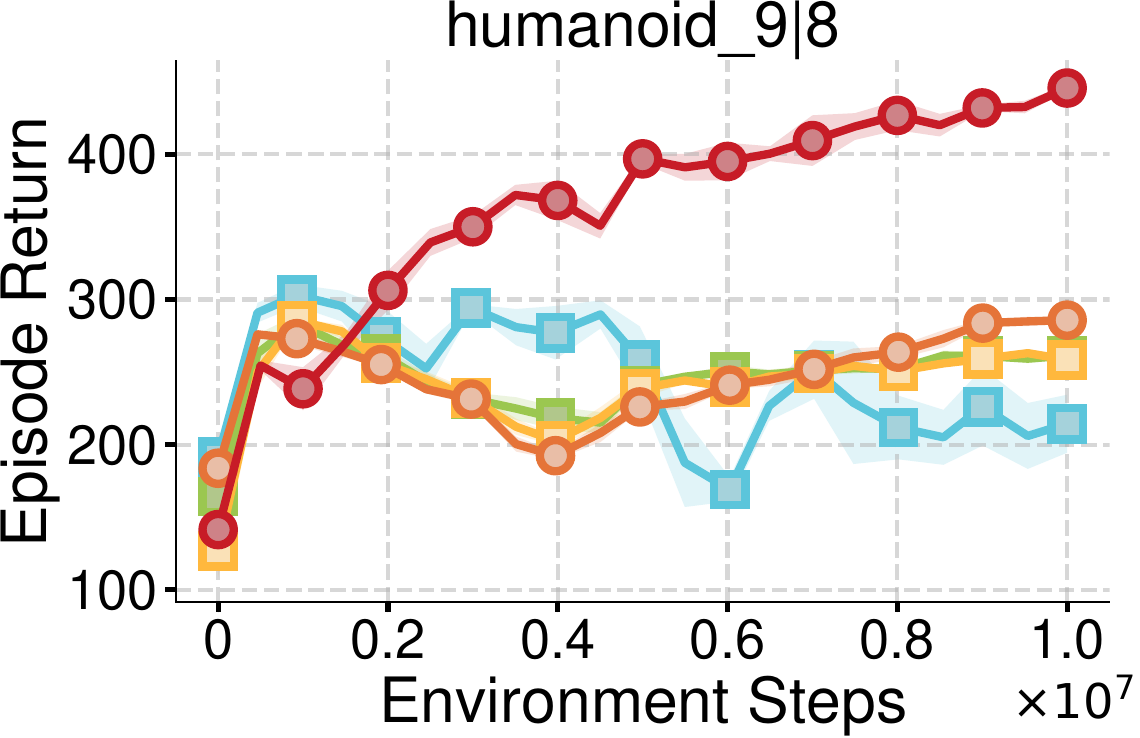}
    
\caption{Figures represent the learning curve of the average win rate performance for 6 tasks on SMAX, and the learning curve of the average episode return performance for 6 tasks on MABrax.}

\Description{Figures represent the learning curve of the average win rate performance for 6 tasks on SMAX, and the learning curve of the average episode return performance for 6 tasks on MABrax.}

\label{fig:onpolicy_smax} 
\end{figure*}

\subsection{Additional Ablations}
To complement the main ablation results on GPAE, we further examine the effect of different truncation strategies when applied to other advantage estimators, namely GAE~\cite{gae} and DAE~\cite{dae}. All variants are trained with the same off-policy sample reuse protocol used in Section~\ref{sec:experiment}, ensuring a consistent comparison.

Table~\ref{tab:appendix_ablation} summarizes the results on SMAX-\texttt{5m\_vs\_6m}. With off-policy corrections, both GAE and DAE show moderate improvements in stability compared to their on-policy counterparts. However, their behaviors under different truncation schemes reveal clear contrasts.  

The consistent advantage of DT-ISR across GAE, DAE, and GPAE demonstrates that applying bounded and credit-sensitive off-policy weights is a general solution for stable multi-agent policy optimization. In all cases, DT-ISR enables off-policy reuse without sacrificing convergence or credit assignment quality.
 
\label{appendix:ablation}
\begin{table}[h]
\centering
    \centering
    \caption{Ablation study on DT-ISR.}
    \begin{tabular}{l|c|c|c}
        \hline\hline
        method & \textbf{GPAE} & GAE & DAE \\
        \hline
        \textbf{DT-ISR} & \textbf{93.7} & 35.1 & 29.4  \\
        ST-ISR & 44.4 & 10.9 & 9.7  \\
        IT-ISR & 58.6 & 9.6 & 14.5 \\
        No Correction & 34.5 & 10.6 & 10.8 \\
        \hline\hline
    \end{tabular}
    \label{tab:appendix_ablation}
\end{table}

\newpage







\newpage
\section{Hyperparameter Settings of the Experiments}

\label{appendix:hyperparameter}

\begin{table}[h]
\centering
\caption{Hyperparameter settings of PPO and Off-Policy PPO in the SMAX and MABrax environment}\vspace{0.5em}
\label{tab:common_parameter_ppo}
\renewcommand{\arraystretch}{1.05} 
\begin{tabular}{l|rr}\hline\hline
\textbf{Parameter} & SMAX & MABrax \\ \hline
\# training timesteps & $1 \times 10^7$ & $1 \times 10^7$ \\
\# parallel environments & $128$ & $64$ \\
\# rollout steps & $128$ & $128$ \\
Adam learning rate & $5\times10^{-4}$ & $1\times10^{-3}$ \\
Anneal learning rate & True & True \\
Update epochs & $5$ & $5$ \\
Clip range & $0.2$ & $0.2$ \\
$\gamma$ & $0.99$ & $0.99$ \\
$\lambda$ & $0.95$ & $0.95$ \\
$\eta$ (GPAE) & $1.05$ & $1.05$ \\
$\beta$ (DAE) & $0.5$ & $0.5$ \\
Entropy coefficient & $0.01$ & $0.01$ \\
Activation & relu & relu \\
\# reuse batches (off-policy) & $4$ & $4$ \\
\# hidden layers & $2$ & $2$ \\
layer width & $128$ & $128$ \\
\hline\hline
\end{tabular}
\end{table}
\begin{table}[h]
\centering
\caption{Hyperparameter settings of QMIX and VDN in the SMAX environment}\vspace{0.5em}
\label{tab:common_parameter_value}
\renewcommand{\arraystretch}{1.05} 
\begin{tabular}{l|r}\hline\hline
\textbf{Parameter} & SMAX \\ \hline
\# training timesteps & $1 \times 10^7$ \\
\# parallel environments & $16$ \\
\# rollout steps & $128$ \\
Adam learning rate & $5\times10^{-5}$ \\
Anneal learning rate & True \\
Buffer size & $5\times10^3$ \\
Buffer batch size& $32$ \\
$\gamma$ & $0.99$ \\
$\epsilon$ (start-decay-finish) & $1.0$-$0.1$-$0.05$ \\
Hidden size & $128$ \\
Update epochs & $8$ \\
\textbf{QMIX specific} &\\
Mixed embedding width & $64$ \\
Mixer hypernet width & $256$ \\
Mixer initial scale & $1\times10^{-3}$ \\
\hline\hline
\end{tabular}
\end{table}